\documentclass{eptcs}

\usepackage{csquotes}
\usepackage{amsmath}
\usepackage{amsthm}
\usepackage{amsfonts}
\usepackage{stmaryrd}
\usepackage{mathtools}
\usepackage{tikz}
\usepackage[spaceRequire=false]{algpseudocodex}
\usepackage{algorithm}
\usepackage{cancel}
\usepackage{thmtools}
\usepackage{thm-restate}
\usepackage[]{hyperref}
\usepackage[]{lineno}
\usepackage{bm}
\usepackage{csvsimple}

\newtheorem{remark}{Remark}
\newtheorem{definition}{Definition}
\newtheorem{example}{Example}

\newtheorem{theorem}{Theorem}

\usepackage{underscore}

\usetikzlibrary{arrows,automata,positioning,petri,calc,backgrounds}

\renewcommand{\textvisiblespace}{\Box}

\newcommand{\loc}{\mathit{loc}}
\newcommand{\var}{\mathit{var}}
\newcommand{\val}{\mathit{val}}
\newcommand{\variables}{\mathit{Vars}}
\newcommand{\values}{\mathit{Values}}
\newcommand{\indexslot}{\mathit{index}}
\newcommand{\states}{Q}
\newcommand{\statevar}{\mathit{state}}
\newcommand{\system}{\mathcal{S}}
\newcommand{\transitions}{\mathbb{T}}
\newcommand{\iterating}{lp}
\newcommand{\local}{lo}
\newcommand{\initial}{\mathit{IVal}}

\newcommand{\ptrs}{\mathit{Ptrs}}
\newcommand{\slots}{\mathit{Sl}}
\newcommand{\slot}{s}
\newcommand{\drop}{\mathit{drop}}
\newcommand{\move}{\mathit{move}}
\newcommand{\origin}{\mathit{origin}}
\newcommand{\target}{\mathit{target}}
\newcommand{\success}{\mathit{succ}}
\newcommand{\fail}{\mathit{fail}}
\newcommand{\self}{\mathit{self}}

\newcommand{\traplang}{\mathcal{L}}

\newcommand{\agents}{\mathbb{A}}
\newcommand{\safetyprop}{\mathcal{P}}

\newcommand{\configuration}{c}
\newcommand{\configurationset}{\mathcal{C}}

\newcommand{\trueval}{\mathit{true}}
\newcommand{\falseval}{\mathit{false}}
\renewcommand{\trueval}{\mathit{\top}}
\renewcommand{\falseval}{\mathit{\bot}}
\newcommand{\stateinitial}{\mathit{initial}}
\newcommand{\stateloop}{\mathit{loop}}
\newcommand{\statebreak}{\mathit{break}}
\newcommand{\statecritical}{\mathit{critical}}
\newcommand{\statedone}{\mathit{done}}

\newcommand{\trap}{O}
\newcommand{\ntrap}{\mathcal{N}}
\newcommand{\ntrapalphabet}{\Sigma_{\system,\agents}}

\newcommand{\set}[1]{\left\{ #1 \right\}}
\newcommand{\tuple}[1]{\left< #1 \right>}
\newcommand{\size}[1]{\left| #1 \right|}
\newcommand{\ali}[1]{\begin{array}{l} #1 \end{array}}
\newcommand{\letter}[1]{\left( \!\!\!\! \begin{array}{cc} #1 \end{array} \!\!\!\! \right)}

\newcommand{\notsqcap}{\cancel{\sqcap}}

\newcommand{\lastrow}{\mathit{lastRow}}
\newcommand{\lastcol}{\mathit{lastCol}}
\newcommand{\firstrow}{\mathit{firstRow}}
\newcommand{\firstcol}{\mathit{firstCol}}
\newcommand{\tile}{\mathit{tile}}
\newcommand{\badflag}{\mathit{badflag}}

\newcommand{\heron}{\texttt{heron}}

\title{Abduction of trap invariants in parameterized systems}
\author{Javier Esparza
\institute{Technical University of Munich}
\email{esparza@in.tum.de}
\and Mikhail Raskin
\institute{Technical University of Munich}
\email{raskin@in.tum.de}
\and Christoph Welzel
\institute{Technical University of Munich}
\email{welzel@in.tum.de}}

\begin{document}


\maketitle

\begin{abstract}
  In a previous paper we have presented a CEGAR approach for the verification
  of parameterized systems with an arbitrary number of processes organized in
  an array or a ring \cite{ERW21}. The technique is based on the iterative computation of
  \emph{parameterized invariants}, i.e., infinite families of invariants for
  the infinitely many instances of the system. Safety properties are proved by
  checking that every global configuration of the system satisfying all
  parameterized invariants also satisfies the property; we have shown that this
  check can be reduced to the satisfiability problem for Monadic Second Order
  on words, which is decidable.

  A strong limitation of the approach is that processes can only have a fixed
  number of variables with a fixed finite range. In particular, they cannot use
  variables with range $[0,N-1]$, where $N$ is the number of processes, which
  appear in many standard distributed algorithms. In this paper, we extend our
  technique to this case. While conducting the check whether a safety property
  is inductive assuming a computed set of invariants becomes undecidable, we
  show how to reduce it to checking satisfiability of a first-order formula. We
  report on experiments showing that automatic first-order theorem provers can
  still perform this check for a collection of non-trivial examples.
  Additionally, we can give small sets of readable invariants for these checks.
\end{abstract}

\section{Introduction}
Many distributed systems consist of an arbitrary number of processes executing
the same algorithm. For every fixed number of processes the system has a finite
state space, and can be verified using conventional model-checking techniques.
However, this technique cannot prove that \emph{all} instances of the system,
one for each number of processes, are correct. Parameterized verification
designs verification procedures for this task. It has developed a number of
techniques, based, among others, on automata theory \cite{AbdullaJNS04},
decidable fragments of first-order logic
\cite{BaukusBLS00,MSZ18,PadonMPSS16}, the theory of well-quasi-orders
\cite{AbdullaCJT96,AbdullaST18,FinkelS01}, or the theory of Vector Addition
Systems \cite{BlondinEH0M20,EsparzaGLM17,GantyM12,GS92}.

In recent work, we have investigated the problem of not only proving that a
parameterized system satisfies a given safety property, but also providing an
explanation of why the property holds in terms of \emph{parameterized
invariants} \cite{BozgaEISW20,BozgaIS21,ERW21}. A parameterized invariant is an
infinite family of invariants, typically one for each number of processes, that
can be finitely described in an adequate language.

Assume for simplicity that an instance of the system consists of a tuple of
processes $\tuple{P_0, \ldots, P_{N-1}}$, each of them with states taken from a
\emph{finite} set of states $Q$. We consider variables as processes; for
example, a boolean variable is a process with states true and false.  A global
state of the instance is a tuple $\tuple{q_0, \ldots, q_{N-1}} \in Q^N$.
Previous research has determined that instances with a small number of
processes can often be proved correct using very simple inductive invariants,
called \emph{trap invariants}. A trap invariant is characterized by a tuple
$\tuple{Q_0, \ldots, Q_{N-1}} \in (2^Q)^N$, called a \emph{trap}. The invariant
states: for every reachable global state $\tuple{q_0, \ldots, q_{N-1}}$ there
exists at least one index $0 \leq i < N$ such that $q_i \in Q_i$. So, loosely
speaking, the invariant says that the trap $\tuple{Q_0, \ldots, Q_{N-1}}$
always remains populated. In the rest of this paper we identify a trap and its
associated invariant if it is clear from the context what we refer to. A set of
traps proves a safety property if every global state satisfying all the
associated trap invariants satisfies the property.

Observe that a trap can be seen as a word of length $N$ over the alphabet
$2^Q$. A \emph{parametric trap} is a regular language over $2^Q$ whose words
are traps. The words of a parametric trap can have arbitrary length, and
therefore induce invariants of different instances. In \cite{ERW21} we have
developed a successful parameterized verification technique, using a CEGAR
loop computing a refinement via traps. The loop maintains a set of parametric
traps, initially empty. Each iteration considers an instance $\tuple{P_0,
\ldots, P_{N-1}}$ of the system, and consists of three steps:
\begin{itemize}
  \item \textbf{Find}. Find traps $T_1, \ldots, T_k$ whose associated
    invariants prove the safety property for $\tuple{P_0, \ldots, P_{N-1}}$. \\
    This is done using finite-state verification techniques.
  \item \textbf{Abduct}. Abduct each trap $T_j$ (or at least some of them) into
    a parametric trap ${\cal T}_j \subseteq (2^Q)^*$. \\ The abduction
    procedure follows from the very simple structure of traps, and its
    correctness is guaranteed by a theorem.
  \item \textbf{Check}. Check whether the current collection of parametric
    traps proves the safety property for \emph{every} instance of the system.
    If not, move to the next instance that cannot be proven correct yet, and
    iterate.

    \noindent In \cite{BozgaEISW20,ERW21} it is shown that this step reduces
    to the satisfiability of a formula of monadic second-order logic on words,
    and is, therefore, decidable.
\end{itemize}

The main restriction of this approach is the fact that the set of states $Q$ of
a process cannot depend on the number $N$ of processes. In particular, this
forbids the use of variables with range $[0, N-1]$, where $N$ is the number of
processes. Using such variables as references to other processes allows to
model \emph{non-atomic} global checks which are essential for many distributed
algorithms \cite{HerlihyS08,Lynch96}; assume, for example, processes that only
have local communication capabilities (as opposed to e.g broadcast
communication) but must implement some kind of mutual exclusion. Such variables
require that a process has to store the values $\set{0, \ldots, N-1}$ in its
state space. In this paper we address particularly the use of iteration
variables for non-atomic global checks. Our contribution is as follows:
\begin{itemize}
  \item We introduce a formal model allowing to describe global checks via
    looping over all processes, inspecting their local states. This happens
    non-atomically, i.e., after inspecting, say, process 2, and before
    inspecting process 3, process 2 may change its local state.
  \item We generalize the CEGAR loop of \cite{ERW21}. For this, we first show
    how to generalize the \textbf{Abduct} step to the new formal model. Second,
    we reduce the \textbf{Check} step to the satisfiability of a formula of a
    certain fragment of first-order logic.
\end{itemize}
The price to pay for the added generality is that the fragment is no longer
decidable. However, we show that with the help of a theorem prover we can
automatically find succinct human-readable correctness proofs for standard
mutual exclusion algorithms.

\paragraph{Related work}
Most similar works on parameterized verification focus on systems communicating by rendez-vous
and systems with atomic global checks. For example, this is the case for
\emph{regular model checking} \cite{AbdullaJNS04}, \emph{parameterized Petri
nets} \cite{ERW21}, \emph{component-based systems} \cite{BozgaEISW20,BozgaIS21}
and \emph{model checking modulo theories} \cite{GhilardiR10}. The latter work
admits unbounded domains using decidable theories but, to the best of our
knowledge, does not explicitly model non-atomic checks. This, however, is the
main focus of the presented work: an explicit extension of the methodology of
\cite{ERW21} to account for \emph{non-atomic} global checks. To this end, we
align our work more with the approaches of \emph{view-abstraction}
\cite{AbdullaHH16} and \emph{symbolic scheme for over-approximations} of
parameterized systems \cite{AbdullaHDR08}. In contrast to these, we
see the benefit of our approach in providing \emph{readable} explanations of
safety proofs in form of first-order formulas. The drawback of our approach is
a significant longer computation time to establish the desired safety
properties. Moreover, the question whether the computed invariants induce the
desired safety property is, in general, undecidable
(cp.~Theorem~\ref{thm:undecidableSatisfiability}).

\section{Formal model}
First, we introduce our model of parameterized systems. Let $\states$ and
$\variables$ be  finite sets of states and variables respectively. Each
variable $\var \in \variables$ ranges over a finite set $\values_{\var}$ of
values. A \emph{valuation} of $\variables$ is a mapping assigning to each $\var
\in \variables$ an element of $\values_{\var}$.

Every $N > 0$ defines an instance of the parameterized system, consisting of
$N$ \emph{agents} with indices $0, 1, \ldots, N-1$. Every agent has a copy of
$Q$ as states, and maintains its own copy of the variables $\variables$. All
agents start in the same initial state $q_0$, with the same initial valuation
$\initial$.

Agents can execute two types of transitions. \emph{Local} transitions are of
the form
\begin{equation}
  \label{eq:local-transition}
  \tuple{\origin, \tuple{\var_{1} \coloneqq \val_{1}, \ldots, \var_{k}
  \coloneqq \val_{k}}, \target}
\end{equation}
where $\origin$ and $\target$ are states, $\var_1, \ldots, \var_{k}$ are
distinct variables, and $\val_{i} \in \values_{\var_{i}}$ for all $1 \leq i
\leq k$. Intuitively, the action allows an agent in state $\origin$ to set its
own variables $\var_1, \ldots, \var_{k}$ to $\val_1, \ldots, \val_{k}$ and move
to state $\target$, all in one single atomic step.

\emph{Loop} transitions allow agents to loop over all agents inspecting their
variables. The inspecting agent or \emph{inspector} first inspects the agent
with index $0$, then the agent with index $1$, and so on, \emph{in different
atomic steps}. Formally, a loop transition has the form
\begin{equation}
  \label{eq:iterating-transition}
  \tuple{\origin, \varphi, \target_{\success}, \target_{\fail}},
\end{equation}
where $\origin$, $\target_{\success}$, and $\target_{\fail}$ are states, and
$\varphi$ is a boolean combination of atoms of the form $\var = \val$ for $\var
\in \variables$ and $\val \in \values_{\var}$ and a predicate $\self$ which
corresponds to the fact that the inspector currently inspects itself. The
inspector conducts a sequence of atomic steps. At the $j$-th step the inspector
checks if the current valuation of the variables of the agent with index $j$
satisfies $\varphi$. If this is the case and $j = n-1$ the inspector moves to
state $\target_{\success}$, if this is the case but $j < n-1$ the inspector
continues to execute the loop transition $t$ and the next time the inspector is
scheduled it checks the agent with index $(j+1)$. Either way, if the agent with
index $j$ does not satisfy $\varphi$ the inspector moves to state
$\target_{\fail}$.

\begin{remark}
\label{rem:loctrans}
  All our results can be easily extended to systems with a third kind of
  transitions of the form $\tuple{\origin, \varphi_l, \varphi_r,
  \target_{\success}, \target_{\fail}}$, allowing an agent to check if the
  variables of its left and right neighbors satisfy $\varphi_l$ and
  $\varphi_r$, respectively, and move to $\target_{\fail}$ or
  $\target_{\success}$ depending on the result. We omit them for simplicity.
\end{remark}

\begin{definition}[Parameterized System]
  A parameterized system over a set $\variables$ of variables is a tuple
  $\system = \tuple{\states, q_{0}, \initial, \transitions_{\local},
  \transitions_{\iterating}}$ where $q_{0} \in \states$ is an \emph{initial
  state}, $\initial$ is an \emph{initial valuation}, and
  $\transitions_{\local}$, $\transitions_{\iterating}$ are finite sets of local
  and loop transitions, respectively.
\end{definition}

\paragraph{Semantics.}  A local configuration of an agent is a triple
consisting of a \emph{location}, a \emph{valuation}, and a \emph{pointer}. The
location is either a state or a loop transition; intuitively, if the
current location is a loop transition $t$, then the agent is currently an
inspector in the \enquote{middle} of executing $t$. The valuation contains
the current values of the agent's variables. The pointer is only important when
the current location is a loop transition $t$; in this case, the pointer
is the index of the agent that the inspector is going to inspect next. With a
small abuse of language, we say that $t$ \emph{points to} this index. A global
configuration, or just configuration for short, is a sequence of local
configurations, one for each agent.

It is convenient to describe global configurations, and the transitions between
them, using a more verbose representation. We first introduce it informally
with the help of an example, a very simplified version of Dijkstra's mutual
exclusion algorithm \cite{Dijkstra2002} \footnote{The version still ensures
mutual exclusion, but offers no progress guarantees. Our experimental results
also include an honest representation of Dijkstra's algorithm.}. Then we give
the formal definition.

\begin{example}[Dijkstra's algorithm for mutual exclusion]
  \label{ex:running-example-spec}
  A reduced version of Dijkstra's algorithm for mutual exclusion can be
  modeled with states $\set{\stateinitial, \stateloop,
  \statebreak, \statecritical, \statedone}$ such that $\stateinitial$ is
  the initial state, and one single boolean variable $b$ with values
  $\set{\trueval, \falseval}$ and initial value $\falseval$. There is one loop
  transition
  \begin{equation*}
    t_{\iterating} = \tuple{\stateloop, \self \lor b = \falseval,
    \statecritical, \statebreak}
  \end{equation*}
  and four local transitions:
  \begin{align*}
    &\tuple{\stateinitial, \tuple{b \coloneqq \trueval}, \stateloop}
    &\tuple{\statebreak, \tuple{b \coloneqq \falseval}, \stateinitial}\\
    &\tuple{\statecritical, \tuple{}, \statedone}
    &\tuple{\statedone, \tuple{b \coloneqq \falseval}, \stateinitial}.
  \end{align*}
\end{example}

Figure \ref{fig:step} describes two steps between configurations of the
instance of Dijkstra's algorithm with three agents. The agents have indices 0,
1, and 2. A local configuration of an agent is represented as a column vector
with five components, labeled $\loc,  \mathit{var}, 0.t_{\iterating},
1.t_{\iterating}, 2.t_{\iterating}$ in the picture. Component $\loc$ gives the
current location of  the agent; for example, in the global configuration at the
top left of Figure~\ref{fig:step}, agent 0 is currently in state
\textit{break}, while agents 1 and 2 are currently executing $t_{\iterating}$.
Component \textit{var} gives the current value of variable $b$. The last three
components give the next process to be inspected. More precisely, if agent $i$
is executing $t_{\iterating}$, and the next agent it will inspect is the one
with index $j$, then we write an $\uparrow$-symbol in the $i.t_{\iterating}$
component of vector $j$. Otherwise, we leave the component blank. In the formal
definitions the absence of $\uparrow$ is denoted with $\textvisiblespace$. A
(global) configuration is a sequence of three vectors, one for each agent. In
the top step of the picture, agent 1 checks variable $b$ of agent 0. Since the
current value of the variable is $\trueval$, the agent moves to the failing
state, which is $\statebreak$. In the bottom step of the picture, agent 1
checks variable $b$ of agent 2. Since the current value is $\falseval$ and $2 =
N-1$, the agent moves to $\statecritical$.

\begin{figure}[t]
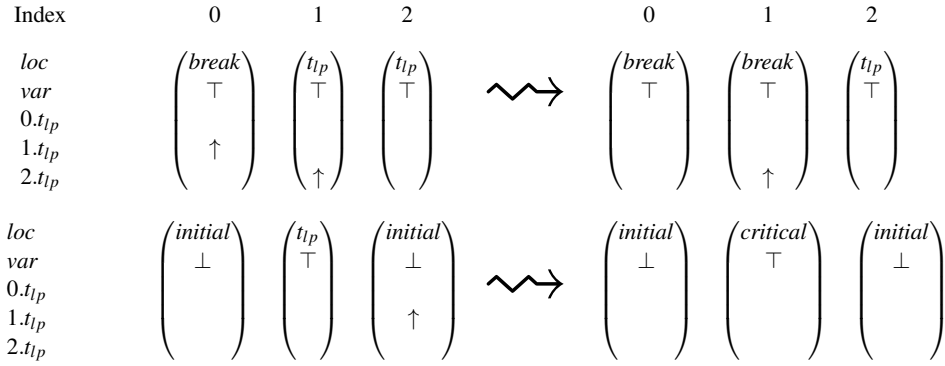

\begin{center}
  \begin{minipage}{0.40\textwidth}
    \begin{center}
      \scalebox{0.8}{\begin{tabular}{ccccccc}
        Index & &  $0$ & $1$ & $2$ \\[0.3cm]
        $\ali{
          \loc \\
          \mathit{var}\\
           0.t_{\iterating} \\
           1.t_{\iterating} \\
           2.t_{\iterating} \\
        }$ & \qquad\qquad &
        $\letter{
          \statebreak\\
          \trueval\\
          \\
          \uparrow\\
          \\
        }$ &
        $\letter{
          t_{\iterating}\\
          \trueval\\
          \\
          \\
          \uparrow\\
        }$ &
        $\letter{
          t_{\iterating}\\
          \trueval\\
          \\
          \\
          \\
        }$
      \end{tabular}}
    \end{center}
  \end{minipage}%
  \begin{minipage}{0.10\textwidth}
    \begin{center}
      \scalebox{3}{$\leadsto$}
    \end{center}
  \end{minipage}%
  \begin{minipage}{0.3\textwidth}
    \begin{center}
      \scalebox{0.8}{\begin{tabular}{ccccc}
        $0$ & $1$ & $2$ &  \\[0.3cm]
        $\letter{
          \statebreak\\
          \trueval\\
          \\
          \\
          \\
        }$ &
        $\letter{
          \statebreak\\
          \trueval\\
          \\
          \\
          \uparrow\\
        }$ &
        $\letter{
          t_{\iterating}\\
          \trueval\\
          \\
          \\
          \\
        }$ &
      \end{tabular}}
    \end{center}
  \end{minipage}%
  \end{center}
  \begin{center}
  \begin{minipage}{0.40\textwidth}
    \begin{center}
      \scalebox{0.8}{\begin{tabular}{ccccccc}
        $\ali{
          \loc \\
          \mathit{var}\\
           0.t_{\iterating} \\
           1.t_{\iterating} \\
           2.t_{\iterating} \\
        }$ & \qquad\qquad &
        $\letter{
          \stateinitial\\
          \falseval\\
          \\
          \\
          \\
        }$ &
        $\letter{
          t_{\iterating}\\
          \trueval\\
          \\
          \\
          \\
        }$ &
        $\letter{
          \stateinitial\\
          \falseval\\
          \\
          \uparrow\\
          \\
        }$
      \end{tabular}}
    \end{center}
  \end{minipage}%
  \begin{minipage}{0.10\textwidth}
    \begin{center}
      \scalebox{3}{$\leadsto$}
    \end{center}
  \end{minipage}%
  \begin{minipage}{0.3\textwidth}
    \begin{center}
      \scalebox{0.8}{\begin{tabular}{ccccc}
        $\letter{
          \stateinitial\\
          \falseval\\
          \\
          \\
          \\
        }$ &
        $\letter{
          \statecritical\\
          \trueval\\
          \\
          \\
          \\
        }$ &
        $\letter{
          \stateinitial\\
          \falseval\\
          \\
          \\
          \\
        }$ &
      \end{tabular}}
    \end{center}
  \end{minipage}%
  \end{center}
  \caption{Steps between configurations of Dijkstra's algorithm for a 3-agent
  instance of Example~\ref{ex:running-example-spec}.}
  \label{fig:step}
  \end{figure}

Formally, configurations are \emph{slotted words}.

\begin{definition}[Slotted words]
  Let $\slots = \set{s_1, \ldots, s_k}$ be a finite set of \emph{slots}, and
  let $\Sigma_{s_1}, \ldots, \Sigma_{s_k}$ be finite alphabets, one for each
  slot. A \emph{slotted word} over $\slots$ is a finite word over the alphabet
  $\Sigma_{s_1} \times \cdots \times \Sigma_{s_k}$.

  Given a slotted word $w = a_{0} \ldots a_{N-1}$ and $0 \leq j < N$, we let
  $a_{j}(s) \in \Sigma_{s}$ denote the component of slot $s$ of $a_{j}$.
\end{definition}

When it is clear from the context, we call a slotted word just a word.

\begin{definition}
  Let $\system = \tuple{\states, q_{0}, \initial, \transitions_{\local},
  \transitions_{\iterating}}$ be a parameterized system over
  $\variables = \{\var_1, \ldots, \var_k\}$, where $\transitions_{\iterating} =
  \{t_1, \ldots, t_\ell \}$, and let $N > 0$.  The set of slots of $\system$
  is
  \begin{equation*}
    \slots_\system = \set{ \loc, \var_1, \ldots, \var_k, 0.t_1, \ldots,
    0.t_\ell, \ldots, (N-1).t_1, \ldots, (N-1).t_\ell }
  \end{equation*}
  with alphabets $\Sigma_\loc = Q \cup \transitions_{\iterating}$;
  $\Sigma_{\var_i} = \values_{\var_i}$ for every $1 \leq i \leq k$; and
  $\Sigma_{i.t_j} = \set{\uparrow, \textvisiblespace}$ for every $0 \leq i < N$
  and $1 \leq j \leq \ell$. With a small abuse of language, we abbreviate the
  description of $\slots_\system$ to
  \begin{equation*}
    \slots_\system = \set{ \loc, \variables,
    0.\transitions_{\iterating}, \ldots,  (N-1).\transitions_{\iterating}}
  \end{equation*}

\smallskip

  A \emph{configuration} of the instance of $\system$ with $N$ agents is a
  slotted word $c_{0} \ldots c_{N-1}$  over $\slots_\system$ satisfying the
  following condition for every $0 \leq i, j < N$ and every $t \in
  \transitions_{\iterating}$: $c_j(i.t) = {\uparrow}$ if{}f $c_i(\loc) = t$ and
  $c_k(i.t) = \textvisiblespace$ for every $k \neq j$. (Intuitively, $c_j(i.t)
  = {\uparrow}$ if{}f agent $i$ is currently executing $t$ and it inspects
  agent $j$ next.) We call each letter $c_i$ a \emph{local configuration}.

  The \emph{initial configuration} of the instance is the word $c_0^N$, where
  $c_0$ is the local configuration given by $c_0(\loc) = q_0$, $c_0(\var_i) =
  \initial(\var_i)$ for every $1 \leq i \leq k$, and $c_0({i.t_j}) =
  \textvisiblespace$ for every $0 \leq i < N$ and $1 \leq j \leq \ell$.
\end{definition}
The semantics of the instances is given by a  transition relation $\vdash$ on
configurations on the basis of the transitions of $\system$. Its formal
definition is routine, and can be found in Appendix~A{ }
of the full version~\cite{esparza2021abduction}.
A configuration $\configuration$ is \emph{reachable} if $\configuration_{0}
\vdash^{*} \configuration$, where $C_0$ is the initial configuration of some
instance.

\section{Analysis of instances}
\label{sec:instance-analysis}
For this section, we fix a parameterized system $\system$ and an instance of
$\system$ of size $N > 1$ as $\tuple{C_{0}, \vdash}$ where $C_{0}$ denotes the
initial configuration, and $\vdash$ its transition relation. Our analysis
relies on techniques originally developed for the analysis of specific
instances of a system modeled as a Petri net. The technique allows one to
compute inductive disjunctive invariants -- so called \emph{traps} -- of the
set of reachable configurations of the instance \cite{EsparzaLMMN14}. In order
to introduce them we first define \emph{powerwords}.

\begin{definition}
  Let $\slots$ be a set of slots with alphabets $\Sigma_{s}$ for every $s \in
  \slots$. A \emph{slotted powerword} is a slotted word over $\slots$ but with
  alphabets $2^{\Sigma_{s}}$ for every $s \in \slots$.

  A word $w = w_{0} \ldots w_{N-1}$ and a powerword $\trap = \trap_{0} \ldots
  \trap_{m-1}$ are \emph{compatible} if $m = n$. Further, $w$ \emph{intersects}
  $\trap$, denoted $w \sqcap \trap$,  if $\trap$ and $w$ are compatible and
  there is an index $0 \leq i < N$ and a slot $s$ such that $w_{i}(s) \in
  \trap_{i}(s)$.
\end{definition}

Intuitively, a trap of this instance $\tuple{C_{0}, \vdash}$ is a
powerword of the instance satisfying the following property: the initial
configuration intersects the trap, and moreover this property is inductive,
i.e., the successor of a configuration intersecting the trap also intersects
the trap.

\begin{definition}[Trap]
\label{def:trap}
  Let $\system$ be a parameterized system, and let $\tuple{C_{0}, \vdash}$ be
  an instance of $\system$. A powerword $\trap$ over $\slots_\system$ is a
  \emph{trap} of $\tuple{C_{0}, \vdash}$ if $C_{0} \sqcap \trap$ and all
  configurations $C, C'$ satisfy: if $C \vdash C'$ and $C \sqcap \trap$, then
  $C' \sqcap \trap$.
\end{definition}

\begin{example}
  \label{ex:trap}
  The following powerword is a trap of the 7-agent-instance of
  Example~\ref{ex:running-example-spec}.
  \begin{center}
    \scalebox{0.8}{\begin{tabular}{ccccccccccccc}
      & &  $0$ & $1$ & $2$ & $3$ & $4$ & $5$ & $6$ \\
      $\ali{
        \mathit{loc} \\
        \mathit{var}\\
         0.t_{\iterating} \\
         1.t_{\iterating} \\
         2.t_{\iterating} \\
         3.t_{\iterating} \\
         4.t_{\iterating} \\
         5.t_{\iterating} \\
         6.t_{\iterating} \\
      }$ & \qquad &
      $\letter{
        \emptyset\\
        \emptyset\\
        \emptyset \\
        \emptyset \\
        \emptyset\\
        \set{\uparrow}\\
        \emptyset\\
        \set{\uparrow}\\
        \emptyset\\
      }$ &
      $\letter{
        \emptyset\\
        \emptyset\\
        \emptyset \\
        \emptyset \\
        \emptyset\\
        \set{\uparrow}\\
        \emptyset\\
        \set{\uparrow}\\
        \emptyset\\
      }$ &
      $\letter{
        \emptyset\\
        \emptyset\\
        \emptyset \\
        \emptyset \\
        \emptyset\\
        \set{\uparrow}\\
        \emptyset\\
        \emptyset\\
        \emptyset\\
      }$ &
      $\letter{
        \set{\statebreak, \stateloop}\\
        \set{\falseval}\\
        \emptyset \\
        \emptyset \\
        \emptyset\\
        \set{\uparrow}\\
        \emptyset\\
        \set{\uparrow}\\
        \emptyset\\
      }$ &
      $\letter{
        \emptyset\\
        \emptyset\\
        \emptyset \\
        \emptyset \\
        \emptyset\\
        \set{\uparrow}\\
        \emptyset\\
        \emptyset\\
        \emptyset\\
      }$  &
      $\letter{
        \set{\statebreak, \stateloop}\\
        \set{\falseval}\\
        \emptyset \\
        \emptyset \\
        \emptyset\\
        \set{\uparrow}\\
        \emptyset\\
        \emptyset\\
        \emptyset\\
      }$ &
      $\letter{
        \emptyset\\
        \emptyset\\
        \emptyset \\
        \emptyset \\
        \emptyset\\
        \emptyset\\
        \emptyset\\
        \emptyset\\
        \emptyset\\
      }$
    \end{tabular}}
  \end{center}
  First observe that the inductiveness
  of intersecting this powerword cannot be changed by any agent but $3$ or $5$
  executing a transition. Now, if either agent $3$ or $5$ executes a local
  transition that sets $\mathit{var}$ to $\bot$ or moves into $\stateloop$ the
  resulting configuration necessarily marks this trap. If agent $3$ ($5$) moves
  from state $\statecritical$ to state $\statedone$ then any configuration that
  intersects this trap before necessarily does not do so at the $\loc$-slot of
  agent $3$ ($5$) and, hence, the resulting configuration still intersects the
  trap. It remains to consider the loop transition. If agent $3$ or $5$ fail
  the loop inspection then either agent moves into state $\statebreak$ which
  necessarily yields a configuration which intersects the trap. If agent $3$
  advances its pointer to some agent in $\set{0,1,2,3,4,5}$
  this yields a configuration intersecting the trap. If agent $3$ advances its
  pointer to agent $6$ then agent $5$ is inspected in this transition. Since
  agent $3$ only advances if $\mathit{var}$ of agent $5$ is currently in state
  $\bot$ the resulting configuration intersects the trap at the
  $\mathit{var}$-slot of agent $5$. If agent $3$ inspects agent $6$ and moves
  to $\statecritical$ only the state of agent $3$ changes from $t_{\iterating}$
  to $\statecritical$ and the $3.t_{\iterating}$-slot of agent $6$. Then,
  however, the resulting configuration intersects the trap at the same slot and
  index as the previous configuration because the occurring change is immaterial
  for the intersection with the trap. Analogous reasoning applies for the
  execution of the loop transition of agent $5$.
\end{example}

Clearly, if $\trap$ is a trap of an instance $\tuple{C_{0}, \vdash}$, then we
have $C' \sqcap \trap$ for every reachable configuration $C'$. Therefore, a set
$I$ of traps induces an over-approximation of the set of reachable
configurations, namely the set of configurations $C$ such that $C \sqcap \trap$
for all $\trap \in I$. Algorithm~\ref{alg:proving-instances} shows a CEGAR
(\emph{counter-example guided abstraction refinement}) loop, adapted from
\cite{EsparzaLMMN14}, to compute a set $I$ of traps proving a given safety
property. Starting with $I = \emptyset$, the algorithm searches for a
configuration (reachable or not) that intersects all traps of $I$ but violates
the safety property. If there is none, the property holds, and the algorithm
terminates; otherwise, the algorithm searches in Line~\ref{lst:line:find-trap}
for a trap that is not intersected by the configuration, and adds it to $I$. It
is well known that the set of such traps can be characterized as the solutions
of a SAT formula.

\begin{algorithm}
  \caption{CEGAR loop to prove configurations unreachable.}
  \label{alg:proving-instances}
  \begin{algorithmic}[1]
    \Require Instance $\tuple{C_{0}, \vdash}$ and collection of
    bad configurations $\mathbb{B}$.
    \Ensure All $B \in \mathbb{B}$ are unreachable from
    $\tuple{C_{0}, \vdash}$ if the result is true.
    \If{$C_{0} \in \mathbb{B}$}
    \State \Return false
    \EndIf
    \State $I \gets \emptyset$
    \While{$C \vdash C'$ exists s.t. $C \sqcap \trap$ for all $\trap \in I$ and
    $C' \in \mathbb{B}$}
    \If{trap $\trap$ exists with $C \notsqcap \trap$}
    \label{lst:line:find-trap}
    \State $I \gets I \cup \set{\trap}$
    \Else
    \State \Return false
    \EndIf
    \EndWhile
    \State \Return true with $I$
  \end{algorithmic}
\end{algorithm}

\section{Parameterized analysis}

In the trap of Example \ref{ex:trap}, the slots  $i.t_{\iterating},$ for $i \in
\set{0,1,2,4,6}$ are \emph{empty}: all their entries are the empty set.
Intuitively, this means that the trap expresses an invariant involving only the
agents 3 and 5. The important fact is that this invariant is \emph{independent}
of how many other agents there are. Indeed, Figure \ref{fig:partraps} is just
one instance of a family of traps for instances of this system. Each element of
this family can be represented as a slotted word with a \emph{fixed} number of
slots, independent of the number of agents. For this we first remove all the
empty slots. Observe, however, that this loses information because we no
longer know which are the indices of the agents of the two nonempty rows for
the loop transitions. So, we also add a new slot \textit{index} with
alphabet $\set{p_0, p_1}$  to mark these two agents. We call this process
\emph{normalization}. The result of normalizing the family at the top of
Figure~\ref{fig:partraps} is shown in the middle of the figure.

The normalized family can be \emph{finitely represented} by the regular
expression at the bottom of the figure. This finite expression can be seen as a
\emph{parameterized invariant} of the system. In the rest of the section we
formalize the notion of a normalized trap, and present a simple \emph{abduction
procedure} that allows us to automatically generalize certain traps into
regular expressions. The procedure can already be sketched in our running
example. Consider, for example, the normalized slotted powerword $\ntrap_{0}
\ldots \ntrap_{N-1}$ shown in the middle of Figure~\ref{fig:partraps} for $2 <
i$ and $i+2 < j < (N-1)-2$. Note that Example~\ref{ex:trap} shows a
conceptually similar trap, although there $j = i + 2$. Regardless, in
Figure~\ref{fig:partraps} we have $\ntrap_{0} = \ntrap_{1} = \ntrap_{2} =
\ldots = \ntrap_{i-1}$ and $\ntrap_{i+1} = \ntrap_{i+2} = \ldots =
\ntrap_{j-1}$. We prove that we can replace  $\ntrap_0\ntrap_1\ntrap_{2}$ by
the regular expression $\ntrap_0^{+}\ntrap_1\ntrap_{2}$, and similarly for
$\ntrap_{i+1}\ntrap_{i+2}\ntrap_{i+3}$, with the guarantee that all words of
the regular expression are traps of the corresponding instances. Indeed, we
prove that any trap with $3$-repetitions of the same letter allows for this
generalization. The complete process of obtaining from an actual trap a
normalized one and, then, a regular expression is sketched in the three lines
of Figure~\ref{fig:partraps}.

\begin{figure}[ht]
  \begin{center}
     \tabcolsep=2pt
     \scalebox{0.7}{\begin{tabular}{lcccccccccccccc}
      & &  $0$ &  & $i-1$ & $i$ & $i+1$ &  & $j-1$ & $j$ & $j+1$ & & $N-1$\\[0.3cm]
      $\ali{
        \mathit{loc} \\
        \mathit{var}\\
         \textcolor{red}{0.t_{\iterating}} \\
         \textcolor{red}{\cdots} \\
         \textcolor{red}{(i-1).t_{\iterating}} \\
         i.t_{\iterating} \\
         \textcolor{red}{(i+1).t_{\iterating}} \\
         \textcolor{red}{\cdots} \\
         \textcolor{red}{(j-1).t_{\iterating}} \\
         j.t_{\iterating} \\
         \textcolor{red}{(j+1).t_{\iterating}}\\
         \textcolor{red}{\cdots} \\
         \textcolor{red}{(N-1).t_{\iterating}}\\
      }$ & \qquad &
      $\letter{
        \emptyset\\
        \emptyset\\
        \textcolor{red}{\emptyset}\\
        \textcolor{red}{\cdots} \\
        \textcolor{red}{\emptyset}\\
        \set{\uparrow}\\
        \textcolor{red}{\emptyset}\\
        \textcolor{red}{\cdots} \\
        \textcolor{red}{\emptyset}\\
        \set{\uparrow}\\
        \textcolor{red}{\emptyset}\\
        \textcolor{red}{\cdots} \\
        \textcolor{red}{\emptyset}\\
      }$ &
      $\begin{array}{c} \\ \\  \cdots \\  \\ \cdots \\  \\  \\ \end{array}$
      &
      $\letter{
        \emptyset\\
        \emptyset\\
        \textcolor{red}{\emptyset}\\
        \textcolor{red}{\cdots} \\
        \textcolor{red}{\emptyset}\\
        \set{\uparrow}\\
        \textcolor{red}{\emptyset}\\
        \textcolor{red}{\cdots} \\
        \textcolor{red}{\emptyset}\\
        \set{\uparrow}\\
        \textcolor{red}{\emptyset}\\
        \textcolor{red}{\cdots} \\
        \textcolor{red}{\emptyset}\\
      }$ &
      $\letter{
        \set{\statebreak, \stateloop}\\
        \set{\falseval}\\
        \textcolor{red}{\emptyset}\\
        \textcolor{red}{\cdots} \\
        \textcolor{red}{\emptyset}\\
        \set{\uparrow}\\
        \textcolor{red}{\emptyset}\\
        \textcolor{red}{\cdots} \\
        \textcolor{red}{\emptyset}\\
        \set{\uparrow}\\
        \textcolor{red}{\emptyset}\\
        \textcolor{red}{\cdots} \\
        \textcolor{red}{\emptyset}\\
      }$ &
      $\letter{
        \emptyset\\
        \emptyset\\
        \textcolor{red}{\emptyset}\\
        \textcolor{red}{\cdots} \\
        \textcolor{red}{\emptyset}\\
        \set{\uparrow}\\
        \textcolor{red}{\emptyset}\\
        \textcolor{red}{\cdots} \\
        \textcolor{red}{\emptyset}\\
        \emptyset\\
        \textcolor{red}{\emptyset}\\
        \textcolor{red}{\cdots} \\
        \textcolor{red}{\emptyset}\\
      }$ &
      $\begin{array}{c} \\ \\  \cdots \\  \\ \cdots \\  \\  \\ \end{array}$
      &
      $\letter{
        \emptyset\\
        \emptyset\\
        \textcolor{red}{\emptyset}\\
        \textcolor{red}{\cdots} \\
        \textcolor{red}{\emptyset}\\
        \set{\uparrow}\\
        \textcolor{red}{\emptyset}\\
        \textcolor{red}{\cdots} \\
        \textcolor{red}{\emptyset}\\
        \emptyset\\
        \textcolor{red}{\emptyset}\\
        \textcolor{red}{\cdots} \\
        \textcolor{red}{\emptyset}\\
      }$ &
      $\letter{
        \set{\statebreak, \stateloop}\\
        \set{\falseval}\\
        \textcolor{red}{\emptyset}\\
        \textcolor{red}{\cdots} \\
        \textcolor{red}{\emptyset}\\
        \set{\uparrow}\\
        \textcolor{red}{\emptyset}\\
        \textcolor{red}{\cdots} \\
        \textcolor{red}{\emptyset}\\
        \emptyset\\
        \textcolor{red}{\emptyset}\\
        \textcolor{red}{\cdots} \\
        \textcolor{red}{\emptyset}\\
      }$ &
      $\letter{
        \emptyset\\
        \emptyset\\
        \textcolor{red}{\emptyset}\\
        \textcolor{red}{\cdots} \\
        \textcolor{red}{\emptyset}\\
        \emptyset\\
        \textcolor{red}{\emptyset}\\
        \textcolor{red}{\cdots} \\
        \textcolor{red}{\emptyset}\\
        \emptyset\\
        \textcolor{red}{\emptyset}\\
        \textcolor{red}{\cdots} \\
        \textcolor{red}{\emptyset}\\
      }$ &
      $\begin{array}{cc} \\ \\  \cdots \\  \\ \cdots \\  \\  \\ \end{array}$
      &
      $\letter{
        \emptyset\\
        \emptyset\\
        \textcolor{red}{\emptyset}\\
        \textcolor{red}{\cdots} \\
        \textcolor{red}{\emptyset}\\
        \emptyset\\
        \textcolor{red}{\emptyset}\\
        \textcolor{red}{\cdots} \\
        \textcolor{red}{\emptyset}\\
        \emptyset\\
        \textcolor{red}{\emptyset}\\
        \textcolor{red}{\cdots} \\
        \textcolor{red}{\emptyset}\\
      }$
      {\ }  \\
      $\ali{
        \indexslot \\
        \mathit{loc} \\
        \mathit{var}\\
         p_0.t_{\iterating} \\
         p_1.t_{\iterating} \\
      }$ & \qquad &
      $\letter{
        \textvisiblespace\\
        \emptyset\\
        \emptyset\\
        \set{\uparrow}\\
        \set{\uparrow}\\
      }$ &
      $\begin{array}{cc} \\ \\  \cdots \\  \\ \cdots \\  \\  \\ \end{array}$
      &
      $\letter{
        \textvisiblespace\\
        \emptyset\\
        \emptyset\\
        \set{\uparrow}\\
        \set{\uparrow}\\
      }$ &
      $\letter{
        p_0\\
        \set{\statebreak, \stateloop}\\
        \set{\falseval}\\
        \set{\uparrow}\\
        \set{\uparrow}\\
      }$ &
      $\letter{
        \textvisiblespace\\
        \emptyset\\
        \emptyset\\
        \set{\uparrow}\\
        \emptyset\\
      }$ &
      $\begin{array}{cc} \\ \\  \cdots \\  \\ \cdots \\  \\  \\ \end{array}$
      &
      $\letter{
        \textvisiblespace\\
        \emptyset\\
        \emptyset\\
        \set{\uparrow}\\
        \emptyset\\
      }$ &
      $\letter{
        p_1\\
        \set{\statebreak, \stateloop}\\
        \set{\falseval}\\
        \set{\uparrow}\\
        \emptyset\\
      }$ &
      $\letter{
        \textvisiblespace\\
        \emptyset\\
        \emptyset\\
        \emptyset\\
        \emptyset\\
      }$ &
      $\begin{array}{cc} \\ \\  \cdots \\  \\ \cdots \\  \\  \\ \end{array}$
      &
      $\letter{
        \textvisiblespace\\
        \emptyset\\
        \emptyset\\
        \emptyset\\
        \emptyset\\
      }$ &
    \end{tabular}}
    {\ } \\
    \scalebox{0.8}{\begin{tabular}{cccccccccccccccc}
      $\letter{
        \textvisiblespace\\
        \emptyset\\
        \emptyset\\
        \set{\uparrow}\\
        \set{\uparrow}\\
      }^{\Huge{+}}$ &
      $\letter{
        \textvisiblespace\\
        \emptyset\\
        \emptyset\\
        \set{\uparrow}\\
        \set{\uparrow}\\
      }$ &
      $\letter{
        \textvisiblespace\\
        \emptyset\\
        \emptyset\\
        \set{\uparrow}\\
        \set{\uparrow}\\
      }$ &
      $\letter{
        p_{0}\\
        \set{\statebreak, \stateloop}\\
        \set{\falseval}\\
        \set{\uparrow}\\
        \set{\uparrow}\\
      }$ &
      $\letter{
        \textvisiblespace\\
        \emptyset\\
        \emptyset\\
        \set{\uparrow}\\
        \emptyset\\
      }^{\Huge{+}}$ &
      $\letter{
        \textvisiblespace\\
        \emptyset\\
        \emptyset\\
        \set{\uparrow}\\
        \emptyset\\
      }$ &
      $\letter{
        \textvisiblespace\\
        \emptyset\\
        \emptyset\\
        \set{\uparrow}\\
        \emptyset\\
      }$ &
      $\letter{
        p_{1}\\
        \set{\statebreak, \stateloop}\\
        \set{\falseval}\\
        \set{\uparrow}\\
        \emptyset\\
      }$  &
      $\letter{
        \textvisiblespace\\
        \emptyset\\
        \emptyset\\
        \emptyset\\
        \emptyset\\
      }^{\Huge{+}}$
      $\letter{
        \textvisiblespace\\
        \emptyset\\
        \emptyset\\
        \emptyset\\
        \emptyset\\
      }$
      $\letter{
        \textvisiblespace\\
        \emptyset\\
        \emptyset\\
        \emptyset\\
        \emptyset\\
      }$
    \end{tabular}}
  \end{center}
\caption{A family of traps, its normalization, and its finite representation as
         a regular expression.}
\label{fig:partraps}
\end{figure}

\subsection{Normalized traps and trap languages}
We now explain how to define sets of traps similar to regular languages. This
requires a finite alphabet which we obtain by only paying attention to the
\enquote{relevant} iteration pointers. The following definition provides the
technical details.
\begin{definition}
  \label{def:norm}
  Let $\system = \tuple{\states, q_{0}, \initial, \transitions_{\local},
  \transitions_{\iterating}}$ be a parameterized system. Let $\agents = \{p_0,
  \ldots, p_{m-1}\}$ be a finite set of \emph{agent names}. A \emph{normalized
  trap} of an instance $\tuple{C_{0}, \vdash}$ over the set of
  agent names $\agents$ is a slotted word $\ntrap_{0} \ldots
  \ntrap_{N-1}$ over the set of slots
  \begin{equation*}
    \set{ \indexslot, \loc, \variables, p_0.\transitions_{\iterating},
    \ldots,  p_{m-1}.\transitions_{\iterating}}
  \end{equation*}
  where the alphabet of $\indexslot$ is $\Sigma_{\indexslot}=\agents
  \cup \set{\textvisiblespace}$, and the other alphabets  are as for normal
  traps, satisfying the following conditions:
  \begin{itemize}
    \item For every name $p_i \in \agents$ there is exactly one letter
      $\ntrap_{j}$ with $\ntrap_{j}(\indexslot) = p_i$. \\ We say that $p_i$ is
      the name of the $j$-th agent.
    \item For every name $p_i \in \agents$ and for every loop transition
      $t_j \in \transitions_{\iterating}$ there is at least one $0 \leq k
      < N$ such that $p_i.t_j(k)= \{\uparrow\}$.
    \item The trap condition, as in Definition \ref{def:trap}.
  \end{itemize}
  We let $\ntrapalphabet$ denote the alphabet of normalized traps over
  $\agents$. A \emph{trap language} over $\ntrapalphabet$ is a regular
  expression of the form $r_0 r_1\ldots r_{\ell-1}$, where $r_i \in \set{a, a^*
  \mid a \in \ntrapalphabet}$ for every $1 \leq i < \ell$.
\end{definition}

The middle part of Figure \ref{fig:partraps} shows a normalized trap over
$\agents = \set{p_0,p_1}$. It is obtained by \emph{normalizing} the trap at the
top of the figure by means of the following procedure:
\begin{itemize}
  \item Remove \enquote{empty} loop slots. Formally, if $i.t_j(k) = \emptyset$
    for every $t_j \in \transitions_{\iterating}$ and $0 \leq k < N$ ,
    remove all the slots of $i. \transitions_{\iterating}$.
  \item Rename all remaining slots $0 \leq i_0, \ldots, i_{m-1} < N$ with
    the \enquote{abstract names} $\set{p_0, \ldots, p_{m-1}}$. Intuitively:
    $p_{k}$ is a placeholder for some index of the word. The corresponding
    concrete index is marked via the additional slot $\indexslot$.
  \item Use the $\indexslot$ slot to give the letters $\ntrap_{i_0}, \ldots,
    \ntrap_{i_{m-1}}$ the names  $p_0, \ldots, p_{m-1}$. Label the others with
    $\textvisiblespace$.
\end{itemize}
The bottom of the figure shows the trap language obtained by normalizing all
the traps shown at the top for all possible values of the indices $i$ and $j$.
The next theorem is the basis of the abduction process that, given a normalized
trap of one instance of the system, produces, if possible, an infinite trap
language of normalized traps.

\begin{theorem}
  \label{thm:main}
  Let $\system$ be a parameterized system. Let $\ntrap_{0} \ldots \ntrap_{N-1}
  \in \ntrapalphabet^{*}$ be a normalized trap of the instance of $\system$
  with $N$ agents. If $\ntrap_i(\indexslot) = \textvisiblespace$ and
  $\ntrap_{i} = \ntrap_{i+1} = \ntrap_{i+2}$, then for every $k \geq 1$ the
  word
  \begin{equation*}
    \ntrap_{0} \ldots \ntrap_{i-1}  \, \ntrap_i^k  \, \ntrap_{i+1} \ldots
    \ntrap_{N-1}
  \end{equation*}
  is a normalized trap of the instance of $\system$ with $N + k$ agents.
\end{theorem}

Before we prove this theorem, we introduce the notions
$\move^{\transitions[\leftarrow]}_{k}$,
$\move^{\transitions[\rightarrow]}_{k}$,
$\drop_{k}$. Conceptually, $\move^{\transitions[\leftarrow]}_{k}$ and
$\move^{\transitions[\rightarrow]}_{k}$ move pointers of the loop
transitions $\transitions \subseteq \transitions_{\iterating}$, which are
currently pointing to $k$, to $k-1$ and $k+1$ respectively. For example
consider
\begin{equation*}
  \small
  \move^{\set{t_{\iterating}}[\leftarrow]}_{1}\left(
  \letter{\statebreak\\
        \trueval\\
        \\
        \uparrow\\
        \\
  }\, \letter{
        t_{\iterating}\\
        \trueval\\
        \\
        \\
        \uparrow\\
  }\, \letter{
        t_{\iterating}\\
        \trueval\\
        \\
        \\
        \\
      }\right) =
  \letter{\statebreak\\
        \trueval\\
        \\
        \uparrow\\
        \uparrow\\
  }\, \letter{
        t_{\iterating}\\
        \trueval\\
        \\
        \\
        \\
  }\, \letter{
        t_{\iterating}\\
        \trueval\\
        \\
        \\
        \\
      }.
\end{equation*}
and, analogously,
\begin{equation*}
  \small
\move^{\set{t_{\iterating}}[\rightarrow]}_{1}\left(
\letter{\statebreak\\
      \trueval\\
      \\
      \uparrow\\
      \\
}\, \letter{
      t_{\iterating}\\
      \trueval\\
      \\
      \\
      \uparrow\\
}\, \letter{
      t_{\iterating}\\
      \trueval\\
      \\
      \\
      \\
    }\right) =
\letter{\statebreak\\
      \trueval\\
      \\
      \uparrow\\
      \\
}\, \letter{
      t_{\iterating}\\
      \trueval\\
      \\
      \\
      \\
}\, \letter{
      t_{\iterating}\\
      \trueval\\
      \\
      \\
      \uparrow\\
    }.
\end{equation*}
Note that $\move^{\transitions[\leftarrow]}_{0}$ is not well-defined and
neither is  $\move^{\transitions[\rightarrow]}_{k}$ if $k$ is the last index
of a configuration since there is no index to move the transition pointers to.

\noindent Formally, we say
$\move^{\transitions[\leftarrow]}_{k}(\configuration_{0} \, \ldots \,
\configuration_{n-1}) = \configuration'_{0} \, \ldots \,
\configuration'_{n-1}$ such that $\configuration'_{\ell} =
\configuration_{\ell}$ for all $\ell \in [n] \setminus \set{k-1, k}$. The
values of variables, the current location and loop transitions that are
not \enquote{moved} do not change: $\configuration'_{k-1}(\slot) =
\configuration_{k-1}(\slot)$ and $\configuration'_{k}(\slot) =
\configuration_{k}(\slot)$ for $\slot \in \variables \cup \set{\loc} \cup
\set{i.t : t \in \transitions_{\iterating}\setminus \transitions \text{ and }i
\in [n]}$. Pointers in $\transitions$, however, move to the left (or to the
right if we consider $\move^{\transitions[\rightarrow]}_{k}$ instead):
$\configuration'_{k}(\ell'.t) = \textvisiblespace$ while
$\configuration'_{k-1}(\ell'.t) = \uparrow$ if either
$\configuration_{k-1}(\ell'.t) = \uparrow$ or $\configuration_{k}(\ell'.t) =
\uparrow$ and otherwise $\configuration'_{k-1}(\ell'.t) = \textvisiblespace$
for all $\ell'\in[n]$ and $t \in \transitions$.

\noindent Secondly, $\drop_{k}$ describes how to remove a specific index $k$
from a configuration. Essentially, we remove all slots
$k.\transitions_{\iterating}$ and the column $k$. As an example, consider
\begin{equation*}
  \small
\drop_{3}\left(
\letter{
  \statebreak\\
  \trueval\\
  \\
  \uparrow\\
  \textcolor{red}{\textvisiblespace}\\
}\, \letter{
  t_{\iterating}\\
  \trueval\\
  \\
  \\
  \textcolor{red}{\uparrow}\\
}\, \letter{
  \textcolor{red}{t_{\iterating}}\\
  \textcolor{red}{\trueval}\\
  \textcolor{red}{\textvisiblespace}\\
  \textcolor{red}{\textvisiblespace}\\
  \textcolor{red}{\textvisiblespace}\\
}\right) =
\letter{
  \statebreak\\
  \trueval\\
  \\
  \uparrow\\
}\, \letter{
  t_{\iterating}\\
  \trueval\\
  \\
  \\
}.
\end{equation*}
Hence, $\drop_{k}(\configuration_{0}
\, \ldots \, \configuration_{n-1}) = \configuration'_{0}
\, \ldots \, \configuration'_{n-2}$ such that
\begin{itemize}
  \item for $\ell_{1} < k$, $\ell_{2} < k$, $t \in \transitions_{\iterating}$
    we have $\configuration'_{\ell_{1}}(\ell_{2}.t)
    = \configuration_{\ell_{1}}(\ell_{2}.t)$ and $\configuration'_{\ell_{1}}
    (\slot) = \configuration_{\ell_{1}}(\slot)$ for all $\slot \in \variables
    \cup \set{\loc}$,
  \item for $\ell_{1} < k$, $k \leq \ell_{2}$, $t \in
    \transitions_{\iterating}$ we have $\configuration'_{\ell_{1}}(\ell_{2}.t)
    = \configuration_{\ell_{1}}(\ell_{2}+1.t)$,
  \item for $k \leq \ell_{1}$, $\ell_{2} < k$,
    $t \in \transitions_{\iterating}$ we have
    $\configuration'_{\ell_{1}}(\ell_{2}.t) =
    \configuration_{\ell_{1}+1}(\ell_{2}.t)$ and $\configuration'_{\ell_{1}}
    (\slot) = \configuration_{\ell_{1}+1}(\slot)$ for all $\slot \in
    \variables \cup \set{\loc}$, and
  \item for $k \leq \ell_{1}$, $k \leq \ell_{2}$, $t \in
    \transitions_{\iterating}$ we have $\configuration'_{\ell_{1}}(\ell_{2}.t)
    = \configuration_{\ell_{1}+1}(\ell_{2}+1.t)$.
\end{itemize}
We make a few observations about $\drop_{k}$:
\begin{itemize}
  \item $\drop_{k}$ yields a configuration if no slot of any loop
    transition of agent $k$ contains the value $\uparrow$.
  \item Let $C$ be a configuration and $\trap$ a compatible
    powerword. If $\configuration_{k}(\slot) \notin \trap_{k}(\slot)$ for all
    slots $\slot$ and $\configuration_{\ell}(k.t) \notin \trap_{\ell}(k.t)$
    for all $\ell$ and $t \in \transitions_{\iterating}$ then
    $\drop_{k}(C) \sqcap \drop_{k}(\trap)$ if{}f
    $C \sqcap \trap$.
\end{itemize}
Equipped with these notions we prove Theorem~\ref{thm:main}:
\begin{proof}
  For the sake of contradiction assume that the statement of
  Theorem~\ref{thm:main} is incorrect.
  Then, we can fix a minimal $k_{0}$ such that the normalized word $\ntrap_{0}
  \ldots \ntrap_{i-1}  \, \ntrap_i^{k_{0}} \, \ntrap_{i+1} \ntrap_{N-1}$ gives
  a powerword $\trap$ which is \emph{not} a trap in the instance of
  $N + k_{0} - 1$ agents.

  Therefore, let $C = \configuration_{0} \, \ldots \,
  \configuration_{N-k_{0}-2}$ and $C'= \configuration'_{0} \, \ldots \,
  \configuration'_{N-k_{0}-2}$ be configurations of the instance with $N + k -
  1$ agents such that $C \vdash C'$ and $C \sqcap \trap$, but $C' \notsqcap
  \trap$.

  Now, observe that $k_{0} > 1$ since $k_{0} = 1$ immediately contradicts with
  the prerequisites of the theorem. However, since $k_{0}$ is minimal, $\ntrap'
  = \ntrap_{0} \ldots \ntrap_{i-1}  \, \ntrap_i^{k_{0}-1}  \, \ntrap_{i+1}
  \ntrap_{N-1}$ \emph{is} a normalized trap of $N + k_{0} - 2$ agents. Let
  $\trap'$ be the corresponding powerword -- which is, in fact, a trap. Note
  that $\trap' = \drop_{i}(\trap) = \drop_{i+1}(\trap) = \ldots =
  \drop_{i+k_{0}}(\trap)$.

  Consider the case where $C \vdash C'$ is an
  instance of a local transition $t$. Thus, $C$ and
  $C'$ differ at exactly one index $j$ and there at most in slots
  from $\variables \cup \set{\loc}$. Pick now $m \in \set{i, i+1, i+2, i+3}
  \setminus \set{j}$. W.l.o.g. we assume that $m-1 \in \set{i, i+1, i+2, i+3}$
  (otherwise exchange $m-1$ with $m+1$ and
  $\move_{m}^{\transitions_{\iterating}[\leftarrow]}$ with
  $\move_{m}^{\transitions_{\iterating}[\rightarrow]}$ within the following
  argument). By the definition of instances of local transitions it
  is straightforward to see that
  $\move_{m}^{\transitions_{\iterating}[\leftarrow]}(C) \vdash
  \move_{m}^{\transitions_{\iterating}[\leftarrow]}(C)$. However,
  we still have
  $\move_{m}^{\transitions_{\iterating}[\leftarrow]}(C) \sqcap
  \trap$ but
  $\move_{m}^{\transitions_{\iterating}[\leftarrow]}(C') \notsqcap
  \trap$ since $\trap_{m-1} = \trap_{m}$ and, thus, $\trap_{m-1}(\ell.t) =
  \trap_{m}(\ell.t)$ for all $\ell \in [n + k_{0} - 1]$ and $t \in
  \transitions_{\iterating}$. Since agent $m$ does not contain any $\uparrow$
  values anymore $(\drop_{m} \circ
  \move_{m}^{\transitions_{\iterating}[\leftarrow]})(C)$ and
  $(\drop_{m} \circ
  \move_{m}^{\transitions_{\iterating}[\leftarrow]})(C')$ indeed
  are configurations.

  \noindent Moreover, by instantiating the same local transition
  for agent $j$ or, if $m \leq j$, $j-1$ gives $(\drop_{m} \circ
  \move_{m}^{\transitions_{\iterating}[\leftarrow]})(C) \vdash
  (\drop_{m} \circ
  \move_{m}^{\transitions_{\iterating}[\leftarrow]})(C')$.
  Note that $\trap_{\ell}(m.t) = \emptyset$ for all $t \in
  \transitions_{\iterating}$ since $\ntrap_{m}(\indexslot) = \textvisiblespace$
  and $\configuration_{m}(\slot) \notin \trap_{m}(\slot)$ for all slots $\slot$
  since $\configuration_{m}(\slot) = \configuration'_{m}(\slot)$ because $j
  \neq m$ and $C' \notsqcap \trap$. Thus, by our observations
  about $\drop_{m}$, we have $(\drop_{m} \circ
  \move_{m}^{\transitions_{\iterating}[\leftarrow]})(C) \sqcap
  \trap'$ because $\trap' = \drop_{m}(\trap)$, but $(\drop_{m} \circ
  \move_{m}^{\transitions_{\iterating}[\leftarrow]})(C') \notsqcap
  \trap'$ contradicting that $\trap'$ is a trap.

  Consider the case that $C \vdash C'$ is an
  instance of a loop transition $t$. Then there are indices $j$ and $p$
  such that $j$ is the agent executing the loop transition while $p$ is
  the agent that $j$ currently inspects. If the loop transition fails, $p$
  changes from $C$ to $C'$ only in slot $j.t$ from
  $\uparrow$ to $\textvisiblespace$ while $j$ changes from $C$ to
  $C'$ only in slot $\loc$ from $t$ to $\target_{\fail}$. In this
  case pick $m \in \set{i, i+1, i+2, i+3} \setminus \set{j, p}$. Again, we
  assume $m-1$ in $\set{i, i+1, i+2, i+3}$.
  Using $\trap_{i} = \trap_{i+1} = \trap_{i+2} = \trap_{i+3}$
  and, $m$ not being involved in this instance of $t$ allows us to consider
  $D = \move_{m}^{\transitions_{\iterating}[\leftarrow]}(C) \vdash
  \move_{m}^{\transitions_{\iterating}[\leftarrow]}(C') = D'$
  instead. Exploiting the same principles above -- namely, the fact that the
  $m.t$ slot is empty for every letter in $\trap$ and the agent $m$ not
  changing from $D$ to $D'$ -- allows us to deduce that $\drop_{m}(D) \sqcap
  \trap'$ while $\drop_{m}(D') \notsqcap \trap'$ which contradicts the
  assumption of $k_{0}$ being minimal.

  It remains to consider the case that this transition is an instance of $t$
  for inspector $j$ and inspectee $p$ which is successful. If $p$ is the last
  agent, then $C$ and $C'$ only differ at indices $j$
  and $p$. Hence, the pattern of the before mentioned cases applies again.

  \noindent Otherwise, only agents $j$, $p$, and $p+1$ change. Then, however,
  there is $m \in \set{i, i+1, i+2, i+3} \setminus \set{j, p, p+1}$. Repeating
  the same arguments as above we can show that $\trap'$ cannot be a trap since
  $(\drop_{m} \circ \move_{m}^{\transitions_{\iterating}[\leftarrow]})
  (C) \vdash (\drop_{m} \circ
  \move_{m}^{\transitions_{\iterating}[\leftarrow]})(C')$ is an
  instance of a transition for the instance of size $N + k_{0} -2$.
\end{proof}

In the next section we show how to use a first-order theorem prover to check
whether every configuration of every instance satisfying the invariants of a
family of trap languages also satisfies a given safety property.

\subsection{Proving safety properties}
\label{sec:proving-safety-properties}
Recall that in \cite{ERW21} we consider parameterized systems without loop
transitions, i.e., having only local transitions and the transitions described
in  Remark \ref{rem:loctrans}. For these systems, the problem whether every
global configuration satisfying a family of trap languages also satisfies a
desired safety property can be reduced to the satisfiability
problem for \emph{WS1S}, which is decidable (see also
\cite{BozgaEISW20,BozgaIS21}). Unfortunately, this is no longer the case for
systems with loop transitions.

Let $\system = \tuple{\states, q_{0}, \initial, \transitions_{\local},
\transitions_{\iterating}}$ be a parameterized system. For every $q \in
\states$, let $\configurationset_q$ denote the set of global configurations of
all instances of $\system$ such that at least one agent is in state $q$.
Similarly, for every variable $x$ and every value $v$ of $\values_{x}$, let
$\configurationset_{v,x}$ be the set of global configurations of all instances
of $\system$ such that for at least one agent the variable $x$ has value $v$.
A \emph{safety property} $\safetyprop(C)$ is a Boolean combination of the
predicates $C \in \configurationset_q$ and $C \in
\configurationset_{v,x}$.  
An example safety property could be 
``there is no agent in the state $critical$,
or there is an agent in state $loop$ and no agent such that $b=\top$''.
Let  $\traplang=\set{\traplang_1, \ldots,
\traplang_k}$ be a finite set of trap languages. Abusing language, let
$\traplang(C)$ denote the predicate that holds for a global configuration $C$
if $C \sqcap O$ for every trap $O \in \traplang$ compatible with $C$.

Given $\system, \safetyprop, \traplang$, we want to decide if every
configuration $C$ satisfying $\traplang(C)$ also satisfies $\safetyprop(C)$,
i.e., whether the invariants of $\traplang$ are sufficient to prove the safety
property $\safetyprop$. Since $\traplang(C)$ is an inductive predicate (i.e.,
$\traplang(C)$ and $C \vdash C'$ imply $\traplang(C')$), this is the case if
(1) $\safetyprop$ holds for every initial configuration  of every instance of
$\system$, and (2) $\traplang(C) \wedge \safetyprop(C) \wedge C \vdash C'$
implies $\safetyprop(C')$. This leads to the following definition:

\begin{definition}
  Let $\system, \traplang,\safetyprop$ be a parameterized system, a trap
  language, and a safety property, respectively. The \emph{inductivity problem}
  consists of deciding if $\; (\traplang(C) \wedge \safetyprop(C) \wedge C
  \vdash C') \rightarrow \safetyprop(C') \; $ holds for every two
  configurations $C, C'$ of $\system$.
\end{definition}

\begin{restatable}{theorem}{undecidableSatisfiability}
  \label{thm:undecidableSatisfiability}
  The inductivity problem is undecidable.
\end{restatable}
\begin{proof}[Sketch]
  Any loop transition $t$ implicitly provides a reference from agent
  to agent. We can use both this and the order of agents to enforce grid
  structures within instances. This allows us to give a reduction of the
  problem whether a periodic tiling for a given set of Wang tiles exists to the
  question above. Then, the result follows from the undecidability of finding
  periodic tilings for Wang tiles \cite{DBLP:journals/tcs/Jeandel10}. For
        the formal details refer to Appendix~B{ }
        of the full version~\cite{esparza2021abduction}.
\end{proof}

Due to this result, there is no hope of reducing the safety problem for systems
with loop transitions to the satisfiability problem of a decidable logic. It is
still possible, though, to reduce it to the satisfiability problem of
first-order logic with monadic predicates and one function symbol, and run it
through an automatic first-order theorem prover.  Informally, we must embed
the statement $$\forall C, C'. (\traplang(C) \wedge \safetyprop(C) \wedge C \vdash C') \rightarrow \safetyprop(C')$$
\noindent in first-order logic.

\paragraph{Embedding into first-order logic.}
In \cite{ERW21} we consider configurations that are solely defined by the
values of the local variables and the state of a process. This can be encoded
by a set of monadic second-order variables; namely, we introduce a monadic
second-order variable $\bm{X}^{\var}_{\val}$ for all values $\val$ of all
$\var \in \variables$.

\noindent In the presence of loop transitions, the local configuration also
includes the pointer, whose value ranges from $0$ to $N-1$, where $N$ is the
number of processes. For this we introduce one function symbol $f$. This
function symbol is interpreted by maps from $\set{0, \ldots, N-1}$ to $\set{0,
\ldots, N-1}$. Combining this, we describe a configuration $C$ as
monadic predicates $\bm{X}^{\var}_{\val}$ for every $\val \in \values_{\var}$
for all $\var \in \variables$, $\bm{X}_{q}$ for every $q \in \states$,
$\bm{X}_{t}$ for every $t \in \transitions_{\iterating}$ and a unary function
symbol $f$. It is obvious, however, that not every interpretation of these
monadic predicates and the function symbol $f$ properly describes a
configuration. Only models which ensure that every variable is currently in
exactly one value for every agent, and that every agent is currently in exactly
one location can be considered proper. For this, we introduce the formula
$\eta$ which ensures these properties. Namely, we define
\begin{equation*}
  \small
  \eta = \forall i:
  \left(\begin{aligned}
    &\bigwedge_{\var \in \variables} \bigvee_{\val \in \values}
      \bm{X}^{\var}_{\val}(i) \land \bigwedge_{\val' \in \values \setminus
    \set{\val}} \lnot \bm{X}^{\var}_{\val}(i)\\
    \land & \bigvee_{\val \in Q \cup \transitions_{\iterating}}
      \bm{X}_{\val}(i) \land \bigwedge_{\val' \in Q \cup
    \transitions_{\iterating} \setminus \set{\val}} \lnot \bm{X}_{\val}(i)\\
  \end{aligned}
  \right).
\end{equation*}
$\eta$ enforces that every agent currently is in exactly one location and
sets every variable to exactly one value. We introduce notions which leverage
this observation to increase readability of formulas. Namely, we write
$\var(i) = \val$ and $\loc(i) = q$ ($\loc(i) = t$) to express
$\bm{X}^{\var}_{\val}(i)$ and $\bm{X}_{q}(i)$ ($\bm{X}_{t}(i)$) respectively.
Moreover, we introduce $\var(i) = \var(j)$ as a short form of
$\bigwedge_{\val \in \values_{\var}} \bm{X}^{\var}_{\val}(i) \leftrightarrow
\bm{X}^{\var}_{\val}(j)$ and, as its dual, $\var(i) \neq \var(j)$ expresses
$\lnot (\var(i) = \var(j))$.

To encode the inductivity problem it is necessary to relate two
configurations $C$ and $C'$ such that $C
\vdash C'$. For this, we introduce a primed version of all monadic
predicates and the function symbol $f$ which we use to encode a configuration
to represent $C'$. Additionally, we axiomatize a total linear
order $\leq$ between the constant symbols $0$ and $(N-1)$ in a formula $\psi$.
Since $\psi$'s definition is fairly standard we omit it here for brevity.
Without increasing the expressiveness of the theory of this total linear order,
we introduce $t < t'$ to express $t \neq t' \land t \leq t'$ and $t+1$ as a
term to refer to the immediate successor of $t$ w.r.t. $\leq$ (this requires
that $t$ cannot be interpreted with $N-1$ which is ensured for all
occurrences). As $\psi$, the encoding of $C \vdash
C'$ is now fairly standard. We postulate the existence of a
formula $\tau$ which relates the predicates and the function symbols $f$ and
$f'$ of $C$ and $C'$. One clause in $\tau$, for
example, which encodes the loop transition $t_{\iterating}$ of
Example~\ref{ex:running-example-spec} (cp. Figure~\ref{fig:step}) is
\begin{equation*}
  \begin{aligned}
    &\forall \ell: (\var(\ell) = \var'(\ell))\\
    \land &\exists i: \loc(i) = t\\
    &\land \forall k \neq i: (\loc(k) = \loc'(k) \land f(k) = f'(k))\\
    &\land \left(\begin{aligned}
      (f(i) \neq i \land \var(f(i)) = \trueval) &\rightarrow \loc'(i) =
      \statebreak\\
      \land((f(i) = i \lor \var(f(i)) = \falseval)\land i < (N-1)) &\rightarrow
      \loc'(i) = t \land f'(i) = f(i)+1\\
      \land((f(i) = i \lor \var(f(i)) = \falseval)\land i = (N-1)) &\rightarrow
      \loc'(i) = \statecritical\\
    \end{aligned}\right).
  \end{aligned}
\end{equation*}

Secondly, the inductivity problem relies on only considering configurations
$C$ which intersect \emph{all} compatible normalized traps of the
computed trap languages. The trap language at the bottom of
Figure~\ref{fig:partraps} describes one of these trap languages.
Careful analysis of Example~\ref{ex:running-example-spec}, however, reveals a
stronger invariant\footnote{In fact, this is one of the invariant languages
that our tool \heron{} computes for this example.}. Namely, that all words in
\begin{equation}
  \label{eq:example-heron-language}
  \scalebox{0.8}{\begin{tabular}{ccccccc}
    $\ali{
      \indexslot \\
      \mathit{loc} \\
      \mathit{var}\\
       p_{0}.t_{\iterating} \\
       p_{1}.t_{\iterating} \\
    }$ & \qquad &
    $\letter{
      \textvisiblespace\\
      \emptyset\\
      \emptyset\\
      \set{\uparrow}\\
      \set{\uparrow}\\
    }^{\Huge{*}}$ &
    $\letter{
      p_{0}\\
      \set{\statebreak, \stateloop}\\
      \set{\falseval}\\
      \set{\uparrow}\\
      \set{\uparrow}\\
    }$ &
    $\letter{
      \textvisiblespace\\
      \emptyset\\
      \emptyset\\
      \set{\uparrow}\\
      \emptyset\\
    }^{\Huge{*}}$ &
    $\letter{
      p_{1}\\
      \set{\statebreak, \stateloop}\\
      \set{\falseval}\\
      \set{\uparrow}\\
      \emptyset\\
    }$  &
    $\letter{
      \textvisiblespace\\
      \emptyset\\
      \emptyset\\
      \emptyset\\
      \emptyset\\
    }^{\Huge{*}}$
  \end{tabular}}
\end{equation}
are normalized traps in instances of the corresponding size. On this basis, we
can derive a formula $\Phi = \eta \land \psi \land \varphi$ whose models are
the global configurations $C$ that intersect \emph{all} traps of this language,
and are compatible with $C$. We already discussed the nature of $\eta$ and
$\psi$. It remains to introduce $\varphi$ which restricts models of $\Phi$ to
configurations which intersect every compatible trap of the language.
The traps of the language are characterized by the positions of the abstract
names $p_0$ and $p_1$. So $\varphi$ is of the form \enquote{for every two
indices $0\leq p_0 < p_1 < N$, the (unique) trap of the language
corresponding to these values of $p_0$ and $p_1$ intersects $C$}.
The formula splits into several cases corresponding to having an intersection
of $C$ and the considered trap at some index $i < p_0$, at $p_0$,
between $p_{0}$ and $p_{1}$, or at $p_{1}$. Consider the case where an
intersection occurs at some index $i < p_0$: since the letters to the left of
the letter with index $p_0$ are of the form
$(\textvisiblespace, \emptyset, \emptyset, \set{\uparrow}, \set{\uparrow})$,
this is the case if one of the agents $p_0$ or $p_1$ is currently executing its
loop transition, and its pointer points to $i$. This is captured by the formula
$t(p_{0}) = i$ or $t(p_{1}) = i$. Proceeding like this for the other cases
gives

\begin{equation*}
  \small
  \varphi = \forall p_{0} < p_{1}~:\left(
    \begin{aligned}
      &\exists i < p_{0}~.~t(p_{0}) = i \lor t(p_{1}) = i\\
      \lor &\statevar(p_{0}) = \statebreak \lor \statevar(p_{0}) = \stateloop
      \lor \var(p_{0}) = \falseval \lor t(p_{0}) = p_{0} \lor t(p_{1}) = p_{0}\\
      \lor &\exists p_{0} < i < p_{1}~.~t(p_{0}) = i\\
      \lor &\statevar(p_{1}) = \statebreak \lor \statevar(p_{1}) = \stateloop
      \lor \var(p_{1}) = \falseval \lor t(p_{0}) = p_{1}
    \end{aligned}
    \right).
\end{equation*}
Recall from Definition~\ref{def:norm} that trap languages are defined as
regular expressions of the form $r_0 r_1\ldots r_{\ell-1}$ where every $r_{i}$
is either a single letter; i.e., some $a$, or the arbitrary repetition of a
single letter; i.e., the expression $a^{*}$. Hence, it is straightforward to
define $\varphi$ for a finite collection of trap languages by generalizing the
example above. Again, we avoid giving this definition for brevity.

The inductiveness problem corresponds now to the question whether
\begin{equation*}
  (\Phi \land \tau \land \safetyprop) \rightarrow \safetyprop'
\end{equation*}
is valid where $\safetyprop$ and $\safetyprop'$ are formalizations of the
desired safety property; e.g., $\lnot (\exists i \neq j: \loc(i) =
\statecritical \land \loc(j) \neq \statecritical)$ and $\lnot (\exists i \neq
j: \loc'(i) = \statecritical \land \loc'(j) \neq \statecritical)$ for
Example~\ref{ex:running-example-spec}.

\section{Experimental results}
We implemented the approach we describe in this paper in our tool \heron{}
\cite{heron-gitlab,heron-artifact} which is written in the Python
programming language. We implemented Algorithm~\ref{alg:proving-instances}
using \texttt{clingo} \cite{DBLP:journals/aicom/GebserKKOSS11} and relied for
first-order theorem proving on \textsc{VAMPIRE} \cite{DBLP:conf/cav/KovacsV13}
and \textsc{CVC4} \cite{DBLP:conf/cav/BarrettCDHJKRT11}. Essentially, \heron{}
runs Algorithm~\ref{alg:proving-instances} and normalizes the occurring traps.
In these normalized traps, it identifies repetitions which can be safely
generalized via Theorem~\ref{thm:main}. Additionally, we try to obtain succinct
invariants. That is, languages with small regular expressions. As an example,
we consider the language from Equation~(\ref{eq:example-heron-language}) as
succinct. To this end, once \heron{} encounters a repetition of letters in a
normalized traps it proceeds to shrink and expand this repetition and checks
whether the resulting normalized traps correspond to actual traps. If a
repetition can be expanded to the generalization threshold it can be safely
repeated arbitrarily often. This heuristic is relatively cheap since it simply
relies on checking whether a normalized trap is indeed a trap but allows us to
easily compute succinct but general trap languages.

\paragraph{Examples}
As a benchmark we used classical examples of mutual exclusion algorithms.
Namely, the mutual exclusion algorithm of Dijkstra \cite{Dijkstra2002} (here we
prove the reduced version as described in Example~\ref{ex:running-example-spec}
as well as a precise formulation), Knuth \cite{DBLP:journals/cacm/Knuth66},
de Bruijn \cite{DBLP:journals/cacm/Bruijn67}, and Eisenberg \& McGuire
\cite{DBLP:journals/cacm/EisenbergM72}. For all these examples we can prove
fully automatically that they ensure the mutual exclusion property for their
respective critical sections. To the best of our knowledge these are the first
fully automatic proofs for the algorithms of Knuth, de Bruijn, and Eisenberg \&
McGuire for a model with non-atomic global checks.

However, all but Example~\ref{ex:running-example-spec} require the use of
global pointers to processes; i.e., we introduce a variable with some value
in $[N]$ for every instance of size $N$ and processes can check the current
value of local variables of the process this pointer currently refers to, set
their own identifier as the current value of the pointer, or upon encountering
a certain state in an agent during an iteration setting the pointer to this
agent. We adapt our approach to account for this: for every pointer we
introduce an additional slot with a corresponding alphabet of $\set{\uparrow,
\textvisiblespace}$ and maintain the invariant that every configuration
contains exactly one agent for which a pointer slot holds the value $\uparrow$.
This makes it necessary to broaden the scope of Theorem~\ref{thm:main}.
Essentially, we obtain now that every transition might interact with all
agents a pointer currently points to. Careful examination of the proof of
Theorem~\ref{thm:main} suggests that we mainly rely on the fact that for every
transition only a finite amount of agents are essential. Adding pointer
variables might increase the amount of agents that are considered by a
transition. This is immaterial for the spirit of Theorem~\ref{thm:main} though,
but increases the amount of repetitions necessary to allow for safe
generalization. Fortunately, as long as the necessary invariants happen to be
simple, the direct performance impact is limited 
to checking a few more easily obtained traps for finite instances.
In Appendix~C{ }
of the full version~\cite{esparza2021abduction}
we shortly
discuss the formalization of this observation.

\paragraph{Results}
The details of our experiments can be found in Figure~\ref{fig:results}.
We report in the first column the name of the considered algorithm. The second
column contains the time it took to prove the mutual exclusion property for
this system. In the third column we give the maximal size for which we
instantiated and analyzed the system. The fourth column reports the amount of
traps that were computed for the various instances of the system while the
fifth column gives the amount of trap languages that were actually necessary to
establish the mutual exclusion property. \heron{} axiomatizes for the
inductivity problem a minimal size of the considered configurations which
exceeds the already analyzed instances. This axiom, however, is not used for
the proof of all examples; i.e., all inductivity queries can be solved without
this axiom in comparable time. This means that the reported trap languages
indeed are sufficient to establish the mutual exclusion property \emph{for all
instances}. In the sixth column we report the maximum amount of abstract
indices that occur in any of the trap languages. Finally, the last column
contains the maximum time it took to solve the final first-order problems for
the various transitions. We observe that the proof search does not dominate the
running time, the real bottleneck being the analysis of the finite instances.

\begin{figure}
  \caption{The experimental results for our tool \heron{}.}
  \label{fig:results}
  \begin{center}
    \scalebox{0.9}{
      \begin{tabular}{l||cccccc}
        Algorithm & time (s) & max. N & \# traps & \# trap languages
        & \begin{tabular}{c}max.\\ \# indices\end{tabular}
        & \begin{tabular}{c}max. proving \\ time (s)\end{tabular}
          \\\hline
        Example~\ref{ex:running-example-spec} & 11 & 8 & 36 & 2 & 2 & 1 \\
        Dijkstra's & 22 & 8 & 75 & 5 & 2 & 1 \\
        Knuth's & 194 & 8 & 160 & 7 & 2 & 1 \\
        de Bruijn's & 76 & 8 & 164 & 6 & 2 & 1 \\
        Eisenberg \& McGuire's & 1055 & 9 & 126 & 6 & 2 & 1\\
      \end{tabular}}
  \end{center}
\end{figure}

The information of Figure~\ref{fig:results} are extracted from log files that
\heron{} produces on its execution. We want to stress specifically that
\heron{} does not only automatically prove the desired property but also gives
a detailed and human-readable explanation. This explanation presents itself in
form of the computed traps for considered finite instances as well as the
computed trap languages which suffice to prove the inductivity problem for all
instances which are larger then the analyzed finite instances. These trap
languages can be read and understood by humans, and are for all considered
examples sufficient to prove inductivity of the mutual exclusion property for
all instances of the parameterized system. As a complete example, consider the
language in Equation~(\ref{eq:example-heron-language}) and
\begin{equation}
  \scalebox{0.8}{\begin{tabular}{ccccc}
    $\ali{
      \mathit{index} \\
      \mathit{loc} \\
      \mathit{var}\\
    }$ & \qquad &
    $\letter{
      \textvisiblespace\\
      \emptyset\\
      \emptyset\\
    }^{\Huge{*}}$ &
    $\letter{
      \textvisiblespace\\
      \set{\stateinitial}\\
      \set{\trueval}\\
    }$ &
    $\letter{
      \textvisiblespace\\
      \emptyset\\
      \emptyset\\
    }^{\Huge{*}}$
  \end{tabular}}
\end{equation}
which are the only languages \heron{} computes to prove
Example~\ref{ex:running-example-spec}. Additionally, \heron{} gives a series of
first-order problems in the widely adopted \texttt{TPTP} format \cite{Sut17}
which describe the inductiveness checks of the proven property. This allows a
user to leverage the advanced tool support for first-order theorem proving
\cite{DBLP:conf/cav/BarrettCDHJKRT11,DBLP:conf/ifm/GleissKS19,DBLP:conf/cade/Reger16}
to verify the proof and, simultaneously, better understand the analyzed system.

In fact, the data of Figure~\ref{fig:results} shows that mutual exclusiveness
in the critical section can be established with very few invariants (e.g., in
less than $7$ trap languages). Moreover, all these languages of normalized
traps need at most two references to indices. Hence, the invariants express --
in some sense -- \emph{local} properties. Finally, we want to draw attention to
the fact that first-order theorem proving is extremely efficient once all
necessary invariants are established. This suggests that the presented approach
is practically viable despite Theorem~\ref{thm:undecidableSatisfiability}.

\section{Conclusion}
In our previous work, we showed the value of structural invariants of
parameterized systems for their analysis \cite{BozgaEISW20}. This contribution
is the natural expansion of our abduction principles from \cite{ERW21} to
parameterized systems with non-atomic global checks. Although one sacrifices
the decidability of the inductivity problem in this expansion, the experimental
data suggests that it remains practically viable. This contrasts observations
from \cite{BozgaEISW20,ERW21} where the inductivity problem is decidable albeit
it often fails in practice due to its computational complexity.
Moreover, most of the time is spent on the analysis of the finite instances,
while generalisations of suitable finite invariants can be verified very
quickly. We therefore conclude that this expansion is worthy of consideration
and, in a broader view, establishes analysis of parameterized systems via
abduction of structural invariants as a promising direction. Moreover, we
believe that our work complements existing approaches like \emph{view
abstraction} \cite{AbdullaHH16}. Especially, since the invariants \heron{}
computes for Dijkstra's algorithm can be similarly expressed in the formalism
of view abstraction. However, \heron{} obtains these languages from structural
analysis of instances while in view abstraction one abstracts from the
reachable set of configurations for fixed sizes. Also, \heron{} computes
comparatively very few invariants but needs \emph{significantly more} time to
do so. Thus, future work can focus on improving the analysis of finite
instances via structural properties. Moreover, it is tempting to further
explore the trade-offs between how strong an analysis is and how easy the
considered structural properties can be abducted to the general case. In
particular, a richer set of considered invariants should allow more diverse
communication structures between the agents. Additionally, one can consider
applying abduction principles to proofs of liveness properties in, e.g., Petri
nets \cite{EsparzaM15}.

\nocite{*}
\bibliographystyle{eptcs}
\bibliography{refs.bib}

\appendix

\section{Formal definition of the semantics}
\label{app:semantics}
Here we describe the formal definition of the semantics of instances in our
parameterized system. For this, let $\configuration = \configuration_{0} \ldots
\configuration_{n-1}$, and $\configuration' = \configuration'_{0} \ldots
\configuration'_{n-1}$ be configurations. Additionally, we write
$\configuration_{i}[s \mapsto v]$ for $s \in \slots$ and $v \in \Sigma_{s}$ to
denote a letter equivalent to $\configuration_{i}$ but where slot $s$ now
contains the value $v$.

\paragraph{Local transition}
Let $t$ be a local transition as in Eq.~(\ref{eq:local-transition}). Now fix
any $i \in [n]$ such that $C_{i}(\loc) = \origin$. Then $\configuration
\vdash \configuration'$ if $\configuration_{j} = \configuration'_{j}$ for all
$j \in [n] \setminus \set{i}$ and $\configuration_{i}' =
\configuration_{i}[\var_{1} \mapsto \val_{1}] \ldots [\var_{k} \mapsto
\val_{k}][\loc \mapsto \target]$.

\paragraph{Iterating transition}
Let $t$ be an iterating transition as in Eq.~(\ref{eq:iterating-transition}).
Again, fix any index $i \in [n]$. We distinguish a few different cases:
\begin{itemize}
  \item If $\configuration_{i}(\loc) = \origin$ then
    agent $i$ just starts the iteration of $t$. Consequently, we say
    $\configuration \vdash \configuration'$ if $\configuration_{0}' =
    \configuration_{0}[t.i \mapsto \uparrow]$, $\configuration_{j}' =
    \configuration_{j}[t.i \mapsto \textvisiblespace]$ for all $j \in [n]
    \setminus{0, i}$ and $\configuration_{i}' = \configuration[i][\loc \mapsto
    t]$. Note, that for the special case that $i = 0$ we require
    $\configuration_{0}' = \configuration_{0}[t.i \mapsto \uparrow][\loc
    \mapsto t]$ and $\configuration_{j}' = \configuration_{j}$ for all
    $j \in [n]\setminus\set{0}$ instead.
  \item If $\configuration_{i}(\loc) = t$ then agent
    $i$ already executes the transition $t$. Thus, there exists exactly one
    index $j$ such that $\configuration_{j}(t.i) = \uparrow$. We
    distinguish cases on the value of $j$:
    \begin{itemize}
      \item If $j = n-1$ then agent $i$ changes its location. Either to
        $\target_{\success}$ or $\target_{\fail}$; depending on whether agent
        $j$ satisfies the condition $\varphi$. Hence, $\configuration_{i}' =
        \configuration_{i}[\loc \mapsto \target_{\success}]$ if
        $\configuration_{j}\models \varphi$. Additionally, $\configuration_{j}'
        = \configuration_{j}[t.i \mapsto \textvisiblespace]$. Note that, if
        $j = i$, then predicates $\self$ are interpreted as true otherwise they
        are interpreted as false. In any way, $\configuration_{k}' =
        \configuration_{k}'$ for all $k\in [n]\setminus\set{i, j}$.
        Analogously, we have $\configuration_{i}' =
        \configuration_{i}[\loc \mapsto \target_{\fail}]$ and
        $\configuration_{j}' = \configuration_{j}[t.i \mapsto
        \textvisiblespace]$ if $\configuration_{j}\not\models \varphi$.
      \item If $i \neq j < n-1$ then we distinguish whether agent $j$
        satisfies the condition $\varphi$: if $\configuration_{j}\models
        \varphi$, then $\configuration_{j}' = \configuration_{j}[t.i \mapsto
        \textvisiblespace]$ and $\configuration_{j+1}[t.i \mapsto
        \uparrow]$ while $\configuration_{k}' = \configuration_{k}$ for all
        $k\in [n]\setminus\set{j, j+1}$. On the other hand; i.e., if
        $\configuration_{j}\not\models \varphi$, we have $\configuration_{i}' =
        \configuration_{i}[\loc \mapsto \target_{\fail}]$ and
        $\configuration_{j}' = \configuration_{j}[t.i \mapsto
        \textvisiblespace]$ while $\configuration_{k}' = \configuration_{k}$
        for all $k\in [n]\setminus\set{i, j}$.
    \end{itemize}
\end{itemize}

\section{Proof of Theorem~\ref{thm:undecidableSatisfiability}}
\label{app:undecidability-proof}
\undecidableSatisfiability*

\begin{proof}
  The proof relies on a reduction from existence of a periodic tiling, using
  agent adjacency and iteration progress to model two directions on the grid.

  \begin{definition}[Periodic tiling of Wang tiles]
    Let $C$ be a finite set of colors. A Wang tiles is a mapping
    $\tau\colon\set{N, E, S, W} \to C$. A finite set of Wang tiles is called a
    tile set.

    A periodic tiling of a tile set $\mathbb{T}$ is a function $f \colon [n]
    \times [m] \rightarrow \mathbb{T}$ for some $n$, $m$ such that
    $f(i, j)(E) = f(i, j\oplus_{m}1)(W)$ and $f(i, j)(N) = f(i \oplus_{n} 1,
    j)(S)$.
  \end{definition}

  We choose $\system$ to have the variables $\variables = \set{\badflag,
  \firstrow, \firstcol, \lastrow, \lastcol, \tile}$ where $\badflag$,
  $\firstrow$, $\firstcol$, $\lastrow$, and $\lastcol$ are Boolean variables;
  i.e., the values for these variables are $\trueval$ and $\falseval$.
  Finally, we set $\values_{\tile}$ to $\mathbb{T}$.

  Now, we consider $\system$ of two transitions: one loop transition
  $t = \tuple{q, \bot, q_{\bot}, q_{\bot}}$ where the condition $\bot$ is a
  logically unsatisfiable statement; e.g., $\badflag = \trueval \land \badflag
  = \falseval$, and one local transition $\tuple{q_l, \badflag = \trueval,
  q_l}$.

  First we construct a set of trap languages such that a configuration without
  any agents with state $q_\bot$ marking all the corresponding traps has to
  encode a correct tiling with Wang tiles. For technical reasons, we will
  encode the tiling without using the agent with the index $0$, which will have
  state $q_l$. We consider the indices $1, 2, \ldots, n-1$ to all currently
  execute the loop transition $t$. We consider the relationship that agent $i$
  inspects agent $j$ next in $i$-th execution of $t$ to be a mapping of $i$ to
  $j$. We have now access to two relations between agents: their natural order
  and the mapping induced by $t$. We use these two relations to construct a
  grid of the following form:
  \begin{center}
    \scalebox{0.4}
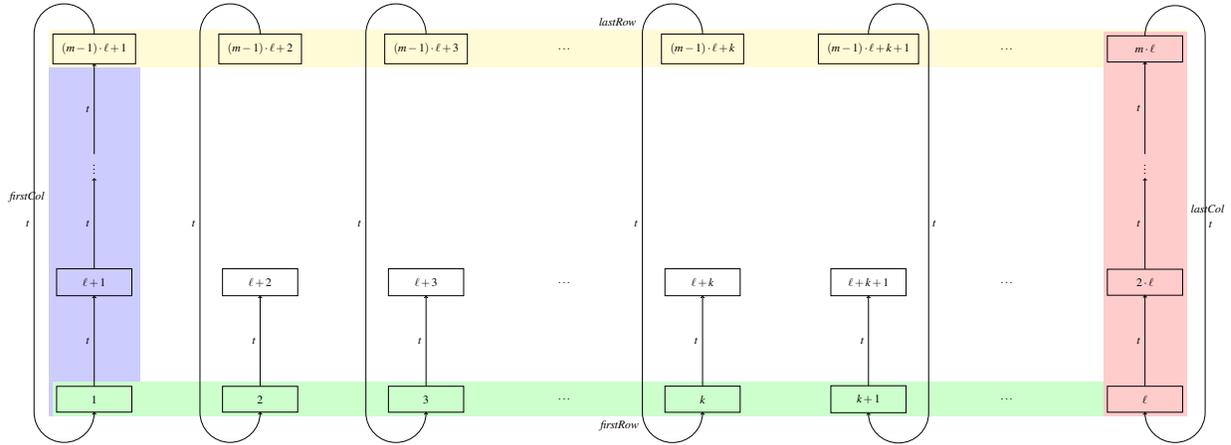
{
      \begin{tikzpicture}[
          node distance=3cm,
          agent/.style={
            rectangle,
            draw,
            inner sep=0.3cm,
            minimum width=2.5cm
          }
        ]
        \node[agent] (1) {$1$};
        \node[right=of 1, agent] (2) {$2$};
        \node[right=of 2, agent] (3) {$3$};
        \node[right=of 3] (4) {$\ldots$};
        \node[right=of 4, agent] (5) {$k$};
        \node[right=of 5, agent] (6) {$k+1$};
        \node[right=of 6] (7) {$\ldots$};
        \node[right=of 7, agent] (8) {$\ell$};

        \node[above=of 1, agent] (l1) {$\ell+1$};
        \node[agent] (l2) at (l1-|2) {$\ell+2$};
        \node[agent] (l3) at (l2-|3) {$\ell+3$};
        \node (l4) at (l3-|4) {$\ldots$};
        \node[agent] (l5) at (l4-|5) {$\ell+k$};
        \node[agent] (l6) at (l5-|6) {$\ell+k+1$};
        \node (l7) at (l6-|7) {$\ldots$};
        \node[agent] (l8) at (l7-|8) {$2\cdot\ell$};

        \node[above=of l1] (dl) {$\vdots$};
        \node (dr) at (dl-|l8) {$\vdots$};

        \node[above=of dl, agent] (m1) {$(m-1)\cdot \ell + 1$};
        \node[agent] (m2) at (m1-|l2) {$(m-1)\cdot \ell + 2$};
        \node[agent] (m3) at (m2-|l3) {$(m-1)\cdot \ell + 3$};
        \node (m4) at (m3-|l4) {$\ldots$};
        \node[agent] (m5) at (m4-|l5) {$(m-1)\cdot \ell + k$};
        \node[agent] (m6) at (m5-|l6) {$(m-1)\cdot \ell + k+1$};
        \node (m7) at (m6-|l7) {$\ldots$};
        \node[agent] (m8) at (m7-|l8) {$m\cdot \ell$};

        \draw[->] (1) to node[left] {$t$} (l1);
        \draw[->] (2) to node[left] {$t$} (l2);
        \draw[->] (3) to node[left] {$t$} (l3);
        \draw[->] (5) to node[left] {$t$} (l5);
        \draw[->] (6) to node[left] {$t$} (l6);
        \draw[->] (8) to node[left] {$t$} (l8);

        \draw[->] (l1) to node[left] {$t$} (dl);
        \draw[->] (l8) to node[left] {$t$} (dr);

        \draw[->] (dl) to node[left] {$t$} (m1);
        \draw[->] (dr) to node[left] {$t$} (m8);

        \path[draw,->] (m1.north) -- (m1.north) arc (0:180:1) -- node[left]
        {$t$} ($ (1.south)-(2,0) $) -- ($ (1.south)-(2,0) $) arc (180:360:1);

        \path[draw,->] (m2.north) -- (m2.north) arc (0:180:1) -- node[left]
        {$t$} ($ (2.south)-(2,0) $) -- ($ (2.south)-(2,0) $) arc (180:360:1);

        \path[draw,->] (m3.north) -- (m3.north) arc (0:180:1) -- node[left]
        {$t$} ($ (3.south)-(2,0) $) -- ($ (3.south)-(2,0) $) arc (180:360:1);

        \path[draw,->] (m5.north) -- (m5.north) arc (0:180:1) -- node[left]
        {$t$} ($ (5.south)-(2,0) $) -- ($ (5.south)-(2,0) $) arc (180:360:1);

        \path[draw,->] (m6.north) -- (m6.north) arc (180:0:1) -- node[right]
        {$t$} ($ (6.south)+(2,0) $) -- ($ (6.south)+(2,0) $) arc (360:180:1);

        \path[draw,->] (m8.north) -- (m8.north) arc (180:0:1) -- node[right]
        {$t$} ($ (8.south)+(2,0) $) -- ($ (8.south)+(2,0) $) arc (360:180:1);

        \begin{scope}[on background layer]
          \begin{scope}[transparency group]
            \begin{scope}[blend mode=multiply]
              \node [rectangle, fill=blue!20, fit=(1) (m1),
                     label={160:$\firstcol$}] {};
              \node [rectangle, fill=green!20, fit=(1) (8),
                     label={270:$\firstrow$}] {};
              \node [rectangle, fill=yellow!20, fit=(m1) (m8),
                     label={90:$\lastrow$}] {};
              \node [rectangle, fill=red!20, fit=(m8) (8),
                     label={10:$\lastcol$}] {};
            \end{scope}
          \end{scope}
        \end{scope}
      \end{tikzpicture}
    }
  \end{center}

  As a set of configurations specified by traps has to be closed under
  transitions, we also include all configurations where at least one agent has
  the state $q_\bot$. We now proceed to define the necessary trap languages and
  prove they ensure the desired structure.

  For convenience we will omit \enquote{or some agent has state $q_\bot$} when
  defining the properties of global configurations. We will also omit the slots
  where all letters have $\emptyset$. All the trap languages we present consist
  of normalized traps, because the transitions can either change $\badflag$ or
  move some agent into the state $q_\bot$; the former never changes whether the
  traps are marked, and the latter ensures that all the traps we present are
  marked.

  First we ensure that the agents have correct locations.
  This is ensured by the following trap languages:
  \begin{center}
    \begin{tabular}{cccc}
      $\letter{
        \textvisiblespace\\
        \set{q_l,q_\bot}\\
      }$&
      $\letter{
        \textvisiblespace\\
        \set{q_\bot}\\
      }^*$&
      \qquad &
      $\ali{
        \indexslot\\
        \loc\\
      }$
    \end{tabular}
  \end{center}
  \begin{center}
    \begin{tabular}{cccccc}
      $\letter{
        \textvisiblespace\\
        \set{q_\bot}\\
        \emptyset\\
      }$&
      $\letter{
        \textvisiblespace\\
        \set{q_\bot}\\
        \set{\uparrow}\\
      }^*$&
      $\letter{
        p_1\\
        \set{q_\bot}\\
        \emptyset\\
      }$&
      $\letter{
        \textvisiblespace\\
        \set{q_\bot}\\
        \set{\uparrow}\\
      }^*$&
      \qquad &
      $\ali{
        \indexslot\\
        \loc\\
        t.p_1\\
      }$
    \end{tabular}
  \end{center}
  It is easy to see that these trap languages require that the first agent
  is in the state $q_l$ while the remaining agents are iterating and not
  examining the agent with the index $0$, not themselves.
  (Or some agent is in the state $q_\bot$)

  For convenience, let us also exclude the case when the two relations
  coincide or are reverse of each other.
  This corresponds to the following trap languages.
  \begin{center}
    \begin{tabular}{ccccccc}
      $\letter{
        \textvisiblespace\\
        \set{q_\bot}\\
        \emptyset\\
      }$&
      $\letter{
        \textvisiblespace\\
        \set{q_\bot}\\
        \emptyset\\
      }^*$&
      $\letter{
        p_1\\
        \set{q_\bot}\\
        \emptyset\\
      }$&
      $\letter{
        \textvisiblespace\\
        \set{q_\bot}\\
        \set{\textvisiblespace}\\
      }$&
      $\letter{
        \textvisiblespace\\
        \set{q_\bot}\\
        \emptyset\\
      }^*$&
      \qquad &
      $\ali{
        \indexslot\\
        \loc\\
        t.p_1\\
      }$
    \end{tabular}
  \end{center}
  \begin{center}
    \begin{tabular}{ccccccc}
      $\letter{
        \textvisiblespace\\
        \set{q_\bot}\\
        \emptyset\\
      }$&
      $\letter{
        \textvisiblespace\\
        \set{q_\bot}\\
        \emptyset\\
      }^*$&
      $\letter{
        \textvisiblespace\\
        \set{q_\bot}\\
        \set{\textvisiblespace}\\
      }$&
      $\letter{
        p_1\\
        \set{q_\bot}\\
        \emptyset\\
      }$&
      $\letter{
        \textvisiblespace\\
        \set{q_\bot}\\
        \emptyset\\
      }^*$&
      \qquad &
      $\ali{
        \indexslot\\
        \loc\\
        t.p_1\\
      }$
    \end{tabular}
  \end{center}
  \begin{center}
    \begin{tabular}{cccccc}
      $\letter{
        \textvisiblespace\\
        \set{q_\bot}\\
        \emptyset\\
      }$&
      $\letter{
        \textvisiblespace\\
        \set{q_\bot}\\
        \set{\textvisiblespace}\\
      }$&
      $\letter{
        \textvisiblespace\\
        \set{q_\bot}\\
        \emptyset\\
      }^*$&
      $\letter{
        p_1\\
        \set{q_\bot}\\
        \emptyset\\
      }$&
      \qquad &
      $\ali{
        \indexslot\\
        \loc\\
        t.p_1\\
      }$
    \end{tabular}
  \end{center}
  \begin{center}
    \begin{tabular}{cccccc}
      $\letter{
        \textvisiblespace\\
        \set{q_\bot}\\
        \emptyset\\
      }$&
      $\letter{
        p_1\\
        \set{q_\bot}\\
        \emptyset\\
      }$&
      $\letter{
        \textvisiblespace\\
        \set{q_\bot}\\
        \emptyset\\
      }^*$&
      $\letter{
        \textvisiblespace\\
        \set{q_\bot}\\
        \set{\textvisiblespace}\\
      }$&
      \qquad &
      $\ali{
        \indexslot\\
        \loc\\
        t.p_1\\
      }$
    \end{tabular}
  \end{center}

  Next we ensure that, roughly speaking, $t(k+1)=t(k)+1$, with wraparound from
  $n$ to $1$. The general case is handled by the following two trap languages.
  \begin{center}
    \begin{tabular}{cccccccccc}
      $\letter{
        \textvisiblespace\\
        \set{q_\bot}\\
        \emptyset\\
        \emptyset\\
      }^*$&
      $\letter{
        p_1\\
        \set{q_\bot}\\
        \emptyset\\
        \emptyset\\
      }$&
      $\letter{
        p_2\\
        \set{q_\bot}\\
        \emptyset\\
        \emptyset\\
      }$&
      $\letter{
        \textvisiblespace\\
        \set{q_\bot}\\
        \emptyset\\
        \emptyset\\
      }^*$&
      $\letter{
        \textvisiblespace\\
        \set{q_\bot}\\
        \set{\textvisiblespace}\\
        \emptyset\\
      }$&
      $\letter{
        \textvisiblespace\\
        \set{q_\bot}\\
        \emptyset\\
        \set{\uparrow}\\
      }$&
      $\letter{
        \textvisiblespace\\
        \set{q_\bot}\\
        \emptyset\\
        \emptyset\\
      }^*$&
      \qquad &
      $\ali{
        \indexslot\\
        \loc\\
        t.p_1\\
        t.p_2\\
      }$
    \end{tabular}
  \end{center}
  \begin{center}
    \begin{tabular}{cccccccccc}
      $\letter{
        \textvisiblespace\\
        \set{q_\bot}\\
        \emptyset\\
        \emptyset\\
      }^*$&
      $\letter{
        \textvisiblespace\\
        \set{q_\bot}\\
        \set{\textvisiblespace}\\
        \emptyset\\
      }$&
      $\letter{
        \textvisiblespace\\
        \set{q_\bot}\\
        \emptyset\\
        \set{\uparrow}\\
      }$&
      $\letter{
        \textvisiblespace\\
        \set{q_\bot}\\
        \emptyset\\
        \emptyset\\
      }^*$&
      $\letter{
        p_1\\
        \set{q_\bot}\\
        \emptyset\\
        \emptyset\\
      }$&
      $\letter{
        p_2\\
        \set{q_\bot}\\
        \emptyset\\
        \emptyset\\
      }$&
      $\letter{
        \textvisiblespace\\
        \set{q_\bot}\\
        \emptyset\\
        \emptyset\\
      }^*$&
      \qquad &
      $\ali{
        \indexslot\\
        \loc\\
        t.p_1\\
        t.p_2\\
      }$
    \end{tabular}
  \end{center}
  We can interpret this language as follows. Given the agents with the indices
  $i$, $i+1$,  $j$, $j+1$, either the agent with the index $i$ is not
  inspecting the agent with the index $j$, or the agent with the index $i+1$ is
  inspecting the agent with the index $j+1$. In other words, if the agent
  number $i$ inspects the agent number $j$, then the agent number $i+1$
  inspects the agent number $j+1$ (assuming $i+1$ and $j+1$ do not exceed
  $n-1$).

  The two special cases correspond to $i$ or $j$ being equal to $n-1$. These
  are described by the following two trap languages.
  \begin{center}
    \begin{tabular}{cccccccccc}
      $\letter{
        \textvisiblespace\\
        \set{q_\bot}\\
        \emptyset\\
        \emptyset\\
      }$&
      $\letter{
        \textvisiblespace\\
        \set{q_\bot}\\
        \emptyset\\
        \set{\uparrow}\\
      }$&
      $\letter{
        \textvisiblespace\\
        \set{q_\bot}\\
        \emptyset\\
        \emptyset\\
      }^*$&
      $\letter{
        p_1\\
        \set{q_\bot}\\
        \emptyset\\
        \emptyset\\
      }$&
      $\letter{
        p_2\\
        \set{q_\bot}\\
        \emptyset\\
        \emptyset\\
      }$&
      $\letter{
        \textvisiblespace\\
        \set{q_\bot}\\
        \emptyset\\
        \emptyset\\
      }^*$&
      $\letter{
        \textvisiblespace\\
        \set{q_\bot}\\
        \set{\textvisiblespace}\\
        \emptyset\\
      }$&
      \qquad &
      $\ali{
        \indexslot\\
        \loc\\
        t.p_1\\
        t.p_2\\
      }$
    \end{tabular}
  \end{center}
  \begin{center}
    \begin{tabular}{cccccccccc}
      $\letter{
        \textvisiblespace\\
        \set{q_\bot}\\
        \emptyset\\
        \emptyset\\
      }$&
      $\letter{
        p_2\\
        \set{q_\bot}\\
        \emptyset\\
        \emptyset\\
      }$&
      $\letter{
        \textvisiblespace\\
        \set{q_\bot}\\
        \emptyset\\
        \emptyset\\
      }^*$&
      $\letter{
        \textvisiblespace\\
        \set{q_\bot}\\
        \set{\textvisiblespace}\\
        \emptyset\\
      }$&
      $\letter{
        \textvisiblespace\\
        \set{q_\bot}\\
        \emptyset\\
        \set{\uparrow}\\
      }$&
      $\letter{
        \textvisiblespace\\
        \set{q_\bot}\\
        \emptyset\\
        \emptyset\\
      }^*$&
      $\letter{
        p_1\\
        \set{q_\bot}\\
        \emptyset\\
        \emptyset\\
      }$&
      \qquad &
      $\ali{
        \indexslot\\
        \loc\\
        t.p_1\\
        t.p_2\\
      }$
    \end{tabular}
  \end{center}

  Given that no agent inspects itself or its neighbour, the trap languages
  defined by now ensure that inspection relation is just addition of a constant
  modulo $n-1$.

  We now proceed to define the trap languages that ensure that the edges of the
  grid are correctly marked. First we consider the first column. Informally
  speaking, the first colour contains the agents reachable from the first one
  by steps corresponding to the inspection relation. First let us ensure that
  all such agents have the $\firstcol$ variable set. The following language
  ensures that the agent with the index $1$ has this variable set.
  \begin{center}
    \begin{tabular}{ccccc}
      $\letter{
        \textvisiblespace\\
        \set{q_\bot}\\
        \emptyset\\
      }$&
      $\letter{
        \textvisiblespace\\
        \set{q_\bot}\\
        \set{\trueval}\\
      }$&
      $\letter{
        \textvisiblespace\\
        \set{q_\bot}\\
        \emptyset\\
      }^*$&
      \qquad &
      $\ali{
        \indexslot\\
        \loc\\
        \firstcol\\
      }$
    \end{tabular}
  \end{center}
  The following two languages ensure correct interactions with the inspection
  relation.
  \begin{center}
    \begin{tabular}{cccccccc}
      $\letter{
        \textvisiblespace\\
        \set{q_\bot}\\
        \emptyset\\
        \emptyset\\
      }$&
      $\letter{
        \textvisiblespace\\
        \set{q_\bot}\\
        \emptyset\\
        \emptyset\\
      }^*$&
      $\letter{
        p_1\\
        \set{q_\bot}\\
        \set{\falseval}\\
        \emptyset\\
      }$&
      $\letter{
        \textvisiblespace\\
        \set{q_\bot}\\
        \emptyset\\
        \emptyset\\
      }^*$&
      $\letter{
        \textvisiblespace\\
        \set{q_\bot}\\
        \set{\trueval}\\
        \set{\textvisiblespace}\\
      }$&
      $\letter{
        \textvisiblespace\\
        \set{q_\bot}\\
        \emptyset\\
        \emptyset\\
      }^*$&
      \qquad &
      $\ali{
        \indexslot\\
        \loc\\
        \firstcol\\
        t.p_1\\
      }$
    \end{tabular}
  \end{center}
  \begin{center}
    \begin{tabular}{cccccccc}
      $\letter{
        \textvisiblespace\\
        \set{q_\bot}\\
        \emptyset\\
        \emptyset\\
      }$&
      $\letter{
        \textvisiblespace\\
        \set{q_\bot}\\
        \emptyset\\
        \emptyset\\
      }^*$&
      $\letter{
        \textvisiblespace\\
        \set{q_\bot}\\
        \set{\trueval}\\
        \set{\textvisiblespace}\\
      }$&
      $\letter{
        \textvisiblespace\\
        \set{q_\bot}\\
        \emptyset\\
        \emptyset\\
      }^*$&
      $\letter{
        p_1\\
        \set{q_\bot}\\
        \set{\falseval}\\
        \emptyset\\
      }$&
      $\letter{
        \textvisiblespace\\
        \set{q_\bot}\\
        \emptyset\\
        \emptyset\\
      }^*$&
      \qquad &
      $\ali{
        \indexslot\\
        \loc\\
        \firstcol\\
        t.p_1\\
      }$
    \end{tabular}
  \end{center}
  What such traps ensure is: for any distinct indices $i$ and $j$,
  either the $i$-th agent does not have $\firstcol$ equal to $\trueval$,
  or it is not currently inspecting the $j$-th agent,
  or $j$-th agent has $\firstcol$ equal to $\trueval$.
  In other words, any agent inspected by an agent with $\firstcol$ set
  also has $\firstcol$ set.
  We now need to ensure that only the agents in the first column have
  the $\firstcol$ variable equal to $\trueval$.
  To do this, we demand that no agent between the agent with the index $1$
  and the agent inspected by the agent with the index $1$ has
  $\firstcol$ set.
  This is ensured by the following trap language.
  \begin{center}
    \begin{tabular}{ccccccccc}
      $\letter{
        \textvisiblespace\\
        \set{q_\bot}\\
        \emptyset\\
        \emptyset\\
      }$&
      $\letter{
        p_1\\
        \set{q_\bot}\\
        \emptyset\\
        \emptyset\\
      }$&
      $\letter{
        \textvisiblespace\\
        \set{q_\bot}\\
        \emptyset\\
        \emptyset\\
      }^*$&
      $\letter{
        \textvisiblespace\\
        \set{q_\bot}\\
        \set{\falseval}\\
        \emptyset\\
      }$&
      $\letter{
        \textvisiblespace\\
        \set{q_\bot}\\
        \emptyset\\
        \emptyset\\
      }^*$&
      $\letter{
        \textvisiblespace\\
        \set{q_\bot}\\
        \emptyset\\
        \set{\textvisiblespace}\\
      }$&
      $\letter{
        \textvisiblespace\\
        \set{q_\bot}\\
        \emptyset\\
        \emptyset\\
      }^*$&
      \qquad &
      $\ali{
        \indexslot\\
        \loc\\
        \firstcol\\
        t.p_1\\
      }$
    \end{tabular}
  \end{center}
  Indeed, using the conversion from disjunction to implication we can say
  that any agent between $1$ and $t(1)$ has $\firstcol$ unset.
  In particular, this implies that $t(1)-1$ divides $n-1$
  (otherwise the wraparound would create a violation).
  The last column can now be defined as the last agent plus all the agents
  immediately before the agents in the first column using the following
  trap languages.
  \begin{center}
    \begin{tabular}{cccc}
      $\letter{
        \textvisiblespace\\
        \set{q_\bot}\\
        \emptyset\\
      }^*$&
      $\letter{
        \textvisiblespace\\
        \set{q_\bot}\\
        \set{\trueval}\\
      }$&
      \qquad &
      $\ali{
        \indexslot\\
        \loc\\
        \lastcol\\
      }$
    \end{tabular}
  \end{center}
  \begin{center}
    \begin{tabular}{cccccc}
      $\letter{
        \textvisiblespace\\
        \set{q_\bot}\\
        \emptyset\\
        \emptyset\\
      }^*$&
      $\letter{
        \textvisiblespace\\
        \set{q_\bot}\\
        \emptyset\\
        \set{\trueval}\\
      }$&
      $\letter{
        \textvisiblespace\\
        \set{q_\bot}\\
        \set{\falseval}\\
        \emptyset\\
      }$&
      $\letter{
        \textvisiblespace\\
        \set{q_\bot}\\
        \emptyset\\
        \emptyset\\
      }^*$&
      \qquad &
      $\ali{
        \indexslot\\
        \loc\\
        \firstcol\\
        \lastcol\\
      }$
    \end{tabular}
  \end{center}
  \begin{center}
    \begin{tabular}{cccccc}
      $\letter{
        \textvisiblespace\\
        \set{q_\bot}\\
        \emptyset\\
        \emptyset\\
      }^*$&
      $\letter{
        \textvisiblespace\\
        \set{q_\bot}\\
        \emptyset\\
        \set{\falseval}\\
      }$&
      $\letter{
        \textvisiblespace\\
        \set{q_\bot}\\
        \set{\trueval}\\
        \emptyset\\
      }$&
      $\letter{
        \textvisiblespace\\
        \set{q_\bot}\\
        \emptyset\\
        \emptyset\\
      }^*$&
      \qquad &
      $\ali{
        \indexslot\\
        \loc\\
        \firstcol\\
        \lastcol\\
      }$
    \end{tabular}
  \end{center}
  The latter two languages say that an agent in the middle has $\lastcol$
  if and only if the next agent has $\firstcol$ set.
  The first row and the last row can now be defined as all the agents
  between the first agent inclusive and its inspectee non-inclusive,
  and between the last agent's inspector non-inclusive and the last agent
  inclusive.
  This corresponds to the following four trap languages.
  \begin{center}
    \begin{tabular}{ccccccccc}
      $\letter{
        \textvisiblespace\\
        \set{q_\bot}\\
        \emptyset\\
        \emptyset\\
      }$&
      $\letter{
        p_1\\
        \set{q_\bot}\\
        \emptyset\\
        \emptyset\\
      }$&
      $\letter{
        \textvisiblespace\\
        \set{q_\bot}\\
        \emptyset\\
        \emptyset\\
      }^*$&
      $\letter{
        \textvisiblespace\\
        \set{q_\bot}\\
        \set{\trueval}\\
        \emptyset\\
      }$&
      $\letter{
        \textvisiblespace\\
        \set{q_\bot}\\
        \emptyset\\
        \emptyset\\
      }^*$&
      $\letter{
        \textvisiblespace\\
        \set{q_\bot}\\
        \emptyset\\
        \set{\textvisiblespace}\\
      }$&
      $\letter{
        \textvisiblespace\\
        \set{q_\bot}\\
        \emptyset\\
        \emptyset\\
      }^*$&
      \qquad &
      $\ali{
        \indexslot\\
        \loc\\
        \firstcol\\
        t.p_1\\
      }$
    \end{tabular}
  \end{center}
  \begin{center}
    \begin{tabular}{ccccccccc}
      $\letter{
        \textvisiblespace\\
        \set{q_\bot}\\
        \emptyset\\
        \emptyset\\
      }$&
      $\letter{
        p_1\\
        \set{q_\bot}\\
        \emptyset\\
        \emptyset\\
      }$&
      $\letter{
        \textvisiblespace\\
        \set{q_\bot}\\
        \emptyset\\
        \emptyset\\
      }^*$&
      $\letter{
        \textvisiblespace\\
        \set{q_\bot}\\
        \emptyset\\
        \set{\textvisiblespace}\\
      }$&
      $\letter{
        \textvisiblespace\\
        \set{q_\bot}\\
        \emptyset\\
        \emptyset\\
      }^*$&
      $\letter{
        \textvisiblespace\\
        \set{q_\bot}\\
        \set{\falseval}\\
        \emptyset\\
      }$&
      $\letter{
        \textvisiblespace\\
        \set{q_\bot}\\
        \emptyset\\
        \emptyset\\
      }^*$&
      \qquad &
      $\ali{
        \indexslot\\
        \loc\\
        \firstcol\\
        t.p_1\\
      }$
    \end{tabular}
  \end{center}
  \begin{center}
    \begin{tabular}{ccccccccc}
      $\letter{
        \textvisiblespace\\
        \set{q_\bot}\\
        \emptyset\\
        \emptyset\\
      }$&
      $\letter{
        \textvisiblespace\\
        \set{q_\bot}\\
        \emptyset\\
        \emptyset\\
      }^*$&
      $\letter{
              p_1\\
        \set{q_\bot}\\
        \emptyset\\
        \emptyset\\
      }$&
      $\letter{
        \textvisiblespace\\
        \set{q_\bot}\\
        \emptyset\\
        \emptyset\\
      }^*$&
      $\letter{
        \textvisiblespace\\
        \set{q_\bot}\\
        \set{\trueval}\\
        \emptyset\\
      }$&
      $\letter{
        \textvisiblespace\\
        \set{q_\bot}\\
        \emptyset\\
        \emptyset\\
      }^*$&
      $\letter{
              \textvisiblespace\\
        \set{q_\bot}\\
        \emptyset\\
            \set{\textvisiblespace}\\
      }$&
      \qquad &
      $\ali{
        \indexslot\\
        \loc\\
        \lastcol\\
        t.p_1\\
      }$
    \end{tabular}
  \end{center}
  \begin{center}
    \begin{tabular}{ccccccccc}
      $\letter{
        \textvisiblespace\\
        \set{q_\bot}\\
        \emptyset\\
        \emptyset\\
      }$&
      $\letter{
        \textvisiblespace\\
        \set{q_\bot}\\
        \emptyset\\
        \emptyset\\
      }^*$&
      $\letter{
        \textvisiblespace\\
        \set{q_\bot}\\
        \set{\falseval}\\
        \emptyset\\
      }$&
      $\letter{
        \textvisiblespace\\
        \set{q_\bot}\\
        \emptyset\\
        \emptyset\\
      }^*$&
      $\letter{
              p_1\\
        \set{q_\bot}\\
        \emptyset\\
        \emptyset\\
      }$&
      $\letter{
        \textvisiblespace\\
        \set{q_\bot}\\
        \emptyset\\
        \emptyset\\
      }^*$&
      $\letter{
              \textvisiblespace\\
        \set{q_\bot}\\
        \emptyset\\
            \set{\textvisiblespace}\\
      }$&
      \qquad &
      $\ali{
        \indexslot\\
        \loc\\
        \lastcol\\
        t.p_1\\
      }$
    \end{tabular}
  \end{center}

  These languages and observations suffice to enforce a grid in the desired
  structure. It remains to assert the compatibility of the placed Wang-tiles.
  For this we write $\overline{\set{\tau}}$ to denote the set $\mathbb{T}
  \setminus \set{\tau}$ and $E(\tau) = \set{\tau' \in \mathbb{T}\mid \tau(E) =
  \tau'(W)}$; i.e., $E(\tau)$ is the set of tiles that can be placed east of
  $\tau$ (similarly for $W(\tau)$, $N(\tau)$ and $S(\tau)$).
  In the following, consider for the following definitions of languages $\tau$
  a parameter and the languages unions for all $\tau \in \mathbb{T}$.
  We start with horizontal compatibility. For two adjacent indices $i$ and
  $i+1$ where agent $i$ does not claim to be on the last column this is
  straightforward:
  \begin{center}
    \begin{tabular}{ccccccc}
      $\letter{
        \emptyset\\
        \set{q_\bot}\\
        \emptyset\\
        \emptyset\\
        \emptyset\\
        \emptyset\\
      }$ &
      $\letter{
        \emptyset\\
        \set{q_\bot}\\
        \emptyset\\
        \emptyset\\
        \emptyset\\
        \emptyset\\
      }^{*}$ &
      $\letter{
        \emptyset\\
        \set{q_\bot}\\
        \emptyset\\
        \emptyset\\
        \set{\trueval}\\
        \overline{\set{\tau}}\\
      }$ &
      $\letter{
        \emptyset\\
        \set{q_\bot}\\
        \emptyset\\
        \emptyset\\
        \emptyset\\
        E(\tau)\\
      }$ &
      $\letter{
        \emptyset\\
        \set{q_\bot}\\
        \emptyset\\
        \emptyset\\
        \emptyset\\
        \emptyset\\
      }^{*}$ &
      \qquad\qquad &
      $\ali{
        \loc \\
        \firstrow\\
        \firstcol\\
        \lastrow\\
        \lastcol\\
        \tile\\
      }$
    \end{tabular}
  \end{center}
  Indeed, what such traps ensure is that whenever the agent with the index $i$
  is not in the last column and it has tile $\tau$, the next agent has one of
  the tiles from $E(\tau)$.
  We only check compatibility from left to right, which is sufficient as long
  as we check each pair of adjacent tiles at least in one direction.

  For the agents from the last column, though, the next agent to the right
  is actually next agent after the corresponding inspector, or the inspector
  of the agent with the next index.
  These two definitions coincide when both are applicable, but in the first row
  and in the last row only one of the two works.
  We thus use the following trap languages.
  \begin{center}
    \begin{tabular}{ccccccccc}
      $\letter{
              \textvisiblespace\\
        \set{q_\bot}\\
        \emptyset\\
        \emptyset\\
        \emptyset\\
        \emptyset\\
        \emptyset\\
        \emptyset\\
      }$ &
      $\letter{
              \textvisiblespace\\
        \set{q_\bot}\\
        \emptyset\\
        \emptyset\\
        \emptyset\\
        \emptyset\\
        \emptyset\\
        \emptyset\\
      }^{*}$ &
      $\letter{
              p_1\\
        \set{q_\bot}\\
        \emptyset\\
        \emptyset\\
        \emptyset\\
        \emptyset\\
        \emptyset\\
        \emptyset\\
      }$ &
      $\letter{
              \textvisiblespace\\
        \set{q_\bot}\\
        \emptyset\\
        \emptyset\\
        \emptyset\\
        \emptyset\\
            E(\tau)\\
        \emptyset\\
      }$ &
      $\letter{
              \textvisiblespace\\
        \set{q_\bot}\\
        \emptyset\\
        \emptyset\\
        \emptyset\\
        \emptyset\\
        \emptyset\\
        \emptyset\\
      }^{*}$ &
      $\letter{
              \textvisiblespace\\
        \set{q_\bot}\\
        \emptyset\\
        \emptyset\\
        \emptyset\\
            \set{\falseval}\\
            \overline{\set{\tau}}\\
        \set{\textvisiblespace}\\
      }$ &
      $\letter{
              \textvisiblespace\\
        \set{q_\bot}\\
        \emptyset\\
        \emptyset\\
        \emptyset\\
        \emptyset\\
        \emptyset\\
        \emptyset\\
      }^{*}$ &
      \qquad &
      $\ali{
              index\\
        \loc \\
        \firstrow\\
        \firstcol\\
        \lastrow\\
        \lastcol\\
        \tile\\
        t.p_1\\
      }$
    \end{tabular}
  \end{center}
  \begin{center}
    \begin{tabular}{ccccccccc}
      $\letter{
              \textvisiblespace\\
        \set{q_\bot}\\
        \emptyset\\
        \emptyset\\
        \emptyset\\
        \emptyset\\
        \emptyset\\
      }$ &
      $\letter{
              \textvisiblespace\\
        \set{q_\bot}\\
        \emptyset\\
        \emptyset\\
        \emptyset\\
        \emptyset\\
        \emptyset\\
        \emptyset\\
      }^{*}$ &
      $\letter{
              p_1\\
        \set{q_\bot}\\
        \emptyset\\
        \emptyset\\
        \emptyset\\
        \emptyset\\
            E(\tau)\\
        \emptyset\\
      }$ &
      $\letter{
              \textvisiblespace\\
        \set{q_\bot}\\
        \emptyset\\
        \emptyset\\
        \emptyset\\
        \emptyset\\
        \emptyset\\
        \emptyset\\
      }^{*}$ &
      $\letter{
              \textvisiblespace\\
        \set{q_\bot}\\
        \emptyset\\
        \emptyset\\
        \emptyset\\
            \set{\falseval}\\
            \overline{\set{\tau}}\\
        \emptyset\\
      }$ &
      $\letter{
              \textvisiblespace\\
        \set{q_\bot}\\
        \emptyset\\
        \emptyset\\
        \emptyset\\
        \emptyset\\
        \emptyset\\
        \set{\textvisiblespace}\\
      }$ &
      $\letter{
              \textvisiblespace\\
        \set{q_\bot}\\
        \emptyset\\
        \emptyset\\
        \emptyset\\
        \emptyset\\
        \emptyset\\
        \emptyset\\
      }^{*}$ &
      \qquad &
      $\ali{
              index\\
        \loc \\
        \firstrow\\
        \firstcol\\
        \lastrow\\
        \lastcol\\
        \tile\\
        t.p_1\\
      }$
    \end{tabular}
  \end{center}
  For vertical compatibility we do a very similar construction.
  \begin{center}
    \begin{tabular}{cccccccc}
      $\letter{
        \textvisiblespace\\
            \set{q_\bot}\\
        \emptyset\\
        \emptyset\\
      }$ &
      $\letter{
        \textvisiblespace\\
            \set{q_\bot}\\
        \emptyset\\
        \emptyset\\
      }^{*}$ &
      $\letter{
        p_{1}\\
            \set{q_\bot}\\
        \overline{\set{\tau}}\\
        \emptyset\\
      }$ &
      $\letter{
        \textvisiblespace\\
            \set{q_\bot}\\
        \emptyset\\
        \emptyset\\
      }^{*}$ &
      $\letter{
        \textvisiblespace\\
            \set{q_\bot}\\
        N(\tau)\\
            \set{\textvisiblespace}\\
      }$ &
      $\letter{
        \textvisiblespace\\
            \set{q_\bot}\\
        \emptyset\\
        \emptyset\\
      }^{*}$ &
      \qquad\qquad &
      $\ali{
        \indexslot\\
        \loc \\
        \tile\\
        t.p_{1}\\
      }$
    \end{tabular}
  \end{center}
  \begin{center}
    \begin{tabular}{cccccccc}
      $\letter{
        \textvisiblespace\\
            \set{q_\bot}\\
        \emptyset\\
        \emptyset\\
      }$ &
      $\letter{
        \textvisiblespace\\
            \set{q_\bot}\\
        \emptyset\\
        \emptyset\\
      }^{*}$ &
      $\letter{
        \textvisiblespace\\
            \set{q_\bot}\\
        N(\tau)\\
            \set{\textvisiblespace}\\
      }$ &
      $\letter{
        \textvisiblespace\\
            \set{q_\bot}\\
        \emptyset\\
        \emptyset\\
      }^{*}$ &
      $\letter{
        p_{1}\\
            \set{q_\bot}\\
        \overline{\set{\tau}}\\
        \emptyset\\
      }$ &
      $\letter{
        \textvisiblespace\\
            \set{q_\bot}\\
        \emptyset\\
        \emptyset\\
      }^{*}$ &
      \qquad\qquad &
      $\ali{
        \indexslot\\
        \loc \\
        \tile\\
        t.p_{1}\\
      }$
    \end{tabular}
  \end{center}

  To summarize our construction so far, we have presented some trap languages
  such that each configuration intersecting with all the traps
  of all the presented
  languages either has an agent in the state $q_\bot$,
  or has one agent in the state $q_l$ and the remaining agents encoding a
  periodic tiling.

  We now define the safety condition $\psi$ to require that either there is at
  least one agent in the state $q_\bot$, or all agents have
  $\badflag=\falseval$. We can also declare the initial state to have all
  agents in the state $q_\bot$. Now the trap languages provide for each $n$ a
  set of initially marked traps with the following property. If there is no
  tiling, these traps ensure that some agent is in the state $q_\bot$, which is
  impossible to leave; thus these traps ensure that the safety condition $\psi$
  is inductive. If there is a tiling, however, we can use it to build a
  configuration where no agent is in the state $q_\bot$, one agent is in the
  state $q_l$, and all agents have $\badflag=\falseval$. This is a safe
  configuration that can however reach an unsafe configuration by using the
  local transition. Thus the question whether some trap languages ensure
  safety, or at least inductiveness of safety, is at least as hard as the
  existence of a periodic tiling using given Wang tiles, and thus undecidable.
\end{proof}

\section{$k$-rendezvousing systems}
\label{app:k-rendezvousing-systems}
On closer inspection of the proof of Theorem~\ref{thm:main} one can see that
the argument mainly relies on the possibility to identify an agent $m$ which is
not involved in the transition and remove it. Pointer variables might cause
transitions to interact with more agents; namely, those agents pointer
variables currently point to. However, since there are only finitely many
pointer variables, the amount of agents a transition might interact with is
still finite. In the following we expand the definitions of
$\move_{m}^{\transitions[\leftarrow]}$ and
$\move_{m}^{\transitions[\rightarrow]}$ to allow pointer slots in the set
$\transitions$. Then, we formalize this in the notion of
\emph{$k$-rendezvousing systems}. Essentially, this means that any transition
only relies on the state of $k$ agents. Formally, we say
\begin{definition}
  We call any parameterized system $\system$ with loop transitions
  $\transitions_{\iterating}$ and pointer variables $\ptrs$
  $k$-rendezvousing if for every instance of any
  transition $\configuration \vdash \configuration'$ there are $k$ indices
  $\set{i_{1}, \ldots, i_{k}}$ such that
  \begin{itemize}
    \item $\configuration_{\ell} = \configuration'_{\ell}$ for all $\ell \in
      [\size{\configuration}] \notin \set{i_{1}, \ldots, i_{k}}$, and
    \item $(\drop_{m} \circ
      \move^{\transitions_{\iterating}\cup\ptrs[\leftarrow]}_{m})
      (\configuration) \vdash
      \move^{\transitions_{\iterating}\cup\ptrs[\leftarrow]}_{m})
      (\configuration')$ is an instance of a transition for every $m \in
      [\size{\configuration}] \setminus \set{0, i_{1}, \ldots, i_{k}}$, and
    \item $(\drop_{m} \circ
      \move^{\transitions_{\iterating}\cup\ptrs[\rightarrow]}_{m})
      (\configuration) \vdash
      \move^{\transitions_{\iterating}\cup\ptrs[\rightarrow]}_{m})
      (\configuration')$ is an instance of a transition for every $m \in
      [\size{\configuration}] \setminus \set{\size{\configuration}-1, i_{1},
      \ldots, i_{k}}$.
  \end{itemize}
\end{definition}

With this definition we can give a straightforward generalization of
Theorem~\ref{thm:main}:
\begin{theorem}
  Let $\system$ be a $k$-rendezvousing parameterized system. Let $\ntrap_{0}
  \ldots \ntrap_{n-1} \in \ntrapalphabet^*$ be a normalized trap of the
  instance of $\system$ with $n$ agents. If $\ntrap_i(\indexslot) =
  \textvisiblespace$ and $\ntrap_{i} = \ntrap_{i+1} = \ldots = \ntrap_{i + k -
  1}$, then for every $\ell \geq 1$ the word
  \begin{equation*}
    \ntrap_{0} \ldots \ntrap_{i-1} \, \ntrap_{i}^{\ell_{0}}  \, \ntrap_{i+1}
    \ldots \ntrap_{n-1}
  \end{equation*}
  is a normalized trap of the instance with $n + \ell - 1$ agents.
\end{theorem}
\begin{proof}
  For the sake of contradiction we assume the statement of the theorem to be
  incorrect. Then we can fix a minimal $\ell_{0} \geq 1$ for which it does not
  hold true anymore. Let $\trap$ be the instance of $\ntrap_{0} \ldots
  \ntrap_{i-1}  \, \ntrap_{i}^{\ell_{0}}  \, \ntrap_{i+1} \, \ldots
  \ntrap_{n-1}$. If $\ell_{0} = 1$ then this already contradicts with the
  assumption of the theorem. Thus, consider the case that $\ell_{0} > 1$. Then,
  however, we have $\trap_{i} = \trap_{i+1} = \ldots = \trap_{i+\ell_{0}+k}$.
  Consequently, there exists $m \in \set{i, \ldots, i+\ell_{0}+k} \setminus
  \set{i_{1}, \ldots, i_{k}}$. Consider the instance $\trap'$ of $\ntrap_{0}
  \ldots \ntrap_{i-1}  \, \ntrap_i^{(\ell_{0}-1)}  \, \ntrap_{i+1} \, \ldots
  \ntrap_{n-1}$ which -- by minimality of $\ell_{0}$ -- is indeed a trap and
  $\drop_{m}(\trap) = \trap'$. By the choice of $m$ there is an adjacent index
  $m'$ of $m$ such that $m' \in \set{i, \ldots, i+\ell_{0}+k}$. W.l.o.g. $m' =
  m-1$ (otherwise exchange
  $\move^{\transitions_{\iterating}\cup\ptrs[\leftarrow]}_{m}$ with
  $\move^{\transitions_{\iterating}\cup\ptrs[\rightarrow]}_{m}$ in the
  following). Since $\trap_{m} = \trap_{m'}$ we have $\configuration \sqcap
  \trap$ if and only if
  $\move^{\transitions_{\iterating}[\leftarrow]\cup\ptrs}_{m}(\configuration)
  \sqcap \trap$ and $\configuration' \sqcap \trap$ if and only if
  $\move^{\transitions_{\iterating}\cup\ptrs[\leftarrow]}_{m}(\configuration')
  \sqcap \trap$. Since $\system$ is $k$-rendezvousing and by choice of $m$ we
  know $(\drop_{m} \circ
  \move^{\transitions_{\iterating}\cup\ptrs[\leftarrow]}_{m})
  (\configuration) \vdash (\drop_{m} \circ
  \move^{\transitions_{\iterating}\cup\ptrs[\leftarrow]}_{m})
  (\configuration')$. Since $\trap$ is not a trap we know $\configuration
  \sqcap \trap$ while $\configuration' \notsqcap \trap$. By our observation
  above this means
  $\move^{\transitions_{\iterating}\cup\ptrs[\leftarrow]}_{m}(\configuration)
  \sqcap \trap$ but
  $\move^{\transitions_{\iterating}\cup\ptrs[\leftarrow]}_{m}(\configuration')
  \notsqcap \trap$. This implies, since $\configuration_{m} =
  \configuration'_{m}$ that $\configuration_{m}(\slot) \notin \trap_{m}(\slot)$
  for all slots $\slot$. Now, we can see that $\trap'$ cannot be a trap since
  $\drop_{m}(\trap) = \trap'$ and, thus, $(\drop_{m} \circ
  \move^{\transitions_{\iterating}\cup\ptrs[\leftarrow]}_{m})(\configuration)
  \sqcap \trap'$ but $(\drop_{m} \circ
  \move^{\transitions_{\iterating}\cup\ptrs[\leftarrow]}_{m})(\configuration')
  \notsqcap \trap'$. This renders $\ell_{0}$ not minimal in contradiction to
  the assumption. The statement of the theorem follows.
\end{proof}

\end{document}